%% file: StructDiv.tex
\documentclass[acmsmall]{acmart}
\pdfoutput=1

\usepackage{balance}  
\usepackage{amsfonts}
\usepackage{url}
\usepackage{algorithm}
\usepackage{algorithmic}
\usepackage{float}
\usepackage{amssymb,amsmath,epsfig,graphics,enumerate,float,subfigure}
\usepackage{graphicx}
\usepackage{color}
\usepackage{xspace}
\usepackage{multirow}

\newcommand{\squishlisttight}{
 \begin{list}{$\bullet$}
  { \setlength{\itemsep}{0pt}
    \setlength{\parsep}{0pt}
    \setlength{\topsep}{0pt}
    \setlength{\partopsep}{0pt}
    \setlength{\leftmargin}{2em}
    \setlength{\labelwidth}{1.5em}
    \setlength{\labelsep}{0.5em}
} }

\newcommand{\squishnumlist} {
\newcounter{qcounter}
\begin{list}{\arabic{qcounter}.~}{\usecounter{qcounter}} 
{  \setlength{\itemsep}{0pt}
    \setlength{\parsep}{0pt}
    \setlength{\topsep}{0pt}
    \setlength{\partopsep}{0pt}
    \setlength{\leftmargin}{2em}
    \setlength{\labelwidth}{1.5em}
    \setlength{\labelsep}{0.5em}
}}

\newcommand{\squishend}{
  \end{list}
}

\newcommand{\eat}[1]{}
\newcommand{\jinbin}[1]{{\color{black}{#1}}}

\newcommand{\kw}[1]{{\ensuremath {\mathsf{#1}}}\xspace}

\newcommand{\stitle}[1]{\vspace{1ex} \noindent{\bf #1}}
\long\def\comment#1{}

\newcommand{\ego}{\kw{ego}-\kw{network}}
\newcommand{\egos}{\kw{ego}-\kw{networks}}

\newcommand{\discomp}{\kw{discriminative} \kw{core}}
\newcommand{\discomps}{\kw{discriminative} \kw{cores}}

\newcommand{\baseline}{\kw{baseline}}
\newcommand{\bound}{\kw{h}-\kw{core}}
\newcommand{\tcore}{\kw{t}-\kw{core}}
\newcommand{\hcore}{\kw{h}-\kw{core}}

\begin{document}

\title{Parameter-free Structural Diversity Search}

\author{Jinbin Huang}
\affiliation{%
  \institution{Hong Kong Baptist University}
  \city{Hong Kong}
  \country{China}}
\email{jbhuang@comp.hkbu.edu.hk}

\author{Xin Huang}
\affiliation{%
  \institution{Hong Kong Baptist University}
  \city{Hong Kong}
  \country{China}}
\email{xinhuang@comp.hkbu.edu.hk}

\author{Yuanyuan Zhu}
\affiliation{%
  \institution{Wuhan University}
  \city{Wuhan}
  \country{China}}
\email{yyzhu@whu.edu.cn}

\author{Jianliang Xu}
\affiliation{%
  \institution{Hong Kong Baptist University}
  \city{Hong Kong}
  \country{China}}
\email{xujl@comp.hkbu.edu.hk}

\renewcommand{\shortauthors}{Huang et al.}





\begin{abstract}
The problem of structural diversity search is to find the top-$k$ vertices with the largest structural diversity in a graph. However, when identifying distinct social contexts, existing structural diversity models (e.g., $t$-sized component, $t$-core, and $t$-brace) are sensitive to an input parameter of $t$. To address this drawback, we propose a parameter-free structural diversity model. Specifically, we propose a novel notation of \discomp, which automatically models various kinds of social contexts without parameter $t$. Leveraging on \discomps and $h$-index, the structural diversity score for a vertex is calculated. We study the problem of parameter-free structural diversity search in this paper. An efficient top-$k$ search algorithm with a well-designed upper bound for pruning is proposed. Extensive experiment results demonstrate the parameter sensitivity of existing $t$-core based model and verify the superiority of our methods.

\end{abstract}

\maketitle

\section{Introduction}\label{sec.intro}
Nowadays, information spreads quickly and widely on social networks (e.g., Twitter, Facebook). Individuals are usually influenced easily by the information received from their social neighborhoods~\cite{hlx2019community}. Recent studies show that social decisions made by individuals often depend on the multiplicity of social contexts inside his/her contact neighborhood, which is termed as \emph{structural diversity}~\cite{UganderBMK12}. Individuals with larger structural diversity, are shown to have higher probability to be affected in the process of social contagion~\cite{UganderBMK12}. Structural diversity search, finding the individuals with the highest structural diversity in graphs, has many applications such as political campaigns~\cite{huckfeldt1995citizens}, viral marketing~\cite{DBLP:conf/kdd/KempeKT03}, promotion of health practices~\cite{UganderBMK12}, facebook user invitations~\cite{UganderBMK12}, and so on.

In the literature, several structural diversity models (e.g., $t$-sized component, $t$-core and $t$-brace) need an input of specific parameter $t$ to model distinct social contexts. A social context is formed by a number of connected users. 
The component-based structural diversity~\cite{UganderBMK12} regards each connected component whose size is larger than $t$ as a social context. Another core-based structural diversity model is defined based on $t$-core. A $t$-core is the largest subgraph such that each vertex has at least $t$ neighbors within $t$-core. The core-based structural diversity model regards each maximal connected $t$-core as a distinct social context. Fig.~\ref{fig.intro} shows the contact neighborhood (\ego) $G_{N(v)}$ of a user $v$. All vertices and edges in \ego $G_{N(v)}$ are in solid lines. Consider the core-based structural diversity model and parameter $t=2$. Subgraphs $H_1$, $H_2$ and $H_3$ are maximal connected $2$-cores. $H_1$, $H_2$, and $H_3$ are regarded as 3 distinct social contexts. Thus, the core-based structural diversity of $v$ is 3. 

\begin{figure}[t]

\centering 
\includegraphics[width=0.4\linewidth]{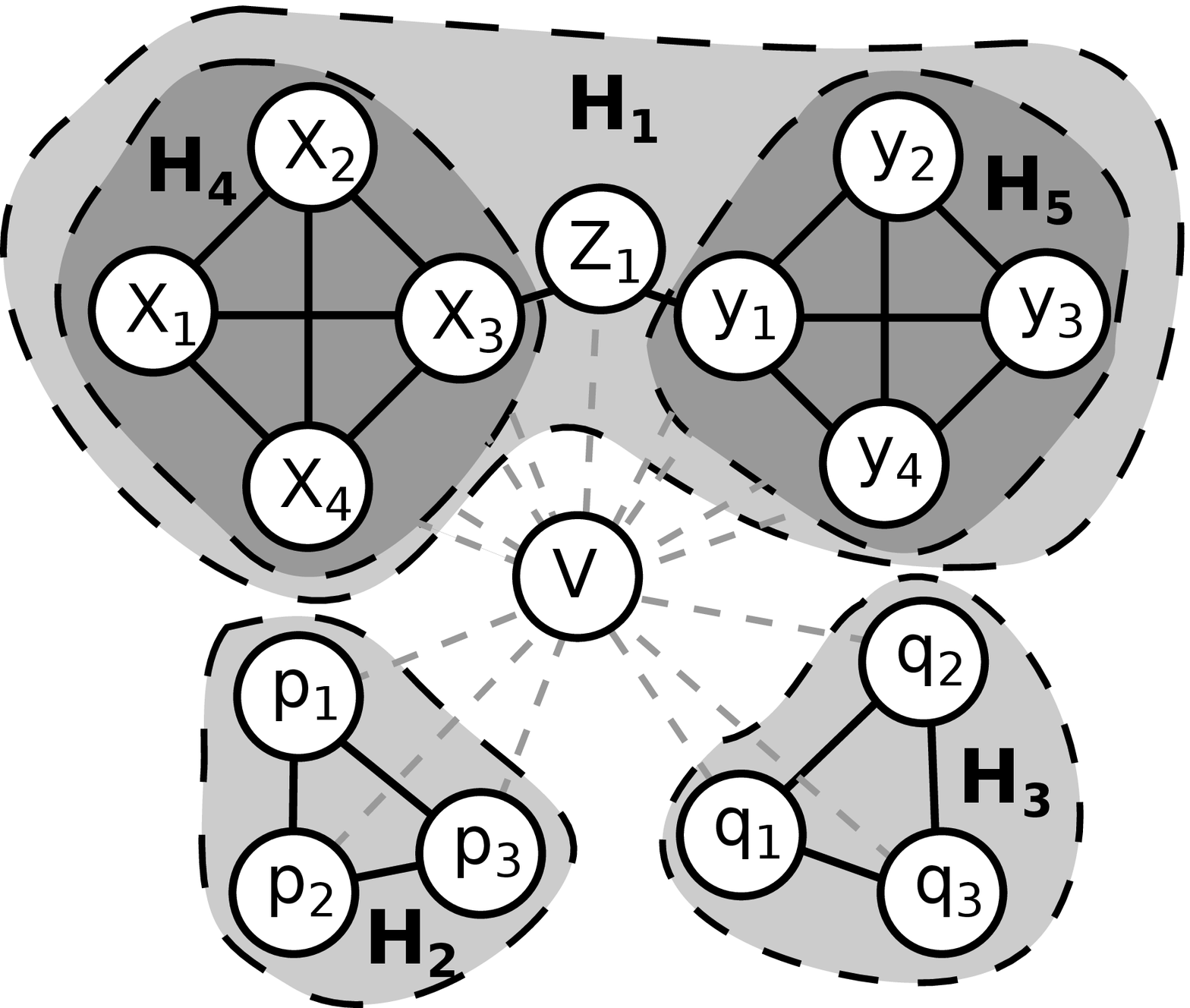}
 \vspace*{-0.2cm}
\caption{The \ego $G_{N(v)}$ of vertex $v$}
\label{fig.intro}
\vspace*{-0.4cm}
\end{figure}

This paper proposes a new parameter-free structural diversity model based on the core-based model~\cite{HuangCLQY15} and h-index measure~\cite{hirsch2005index}. Our parameter-free model does not need the input of parameter $t$ any more. This avoids suffering from the limitations of setting parameter $t$. \eat{We take core-based model as example in the following.}\jinbin{We show two major drawbacks of the $t$-core based model as follows.}

\begin{itemize}
    \vspace{-0.2cm}
    \item \textbf{Sensitivity of t-core based model}. The number of social contexts is sensitive to parameter $t$. On the one hand, if $t$ is set to a large value, it may discard small and weakly-connected social contexts; On the other hand, if $t$ is set to a small value, it may have weak ability of recognizing strongly-connected social contexts fully. Consider the contact neighborhood $G_{N(v)}$ of a user $v$ in Fig.~\ref{fig.intro}. When $t=2$, the structural diversity of $v$ is 3. When $t=3$, $H_2$ and $H_3$ are 2-cores and disqualified for social contexts, due to the requirement of social contexts as 3-core. Meanwhile, $H_1$ is decomposed as two components of 3-core as $H_4$ and $H_5$. Thus, the structural diversity of $v$ becomes 2. However, when $t\geq 4$, the structural diversity of $v$ is 0. This example clearly shows the  sensitivity of structural diversity w.r.t. parameter $t$.

    \item \textbf{Inflexibility of t-core based model}. Structural diversity model lacks flexibility for different vertices using the same parameter $t$. Generally, different social contexts should not be modeled and quantified using the same criteria of parameter $t$.  For example, in a social network, the social contexts of a famous singer and a junior student can be dramatically different in terms of size and density. Thus, it is difficult to choose one consistent value $t$ for different vertices in a graph. In Fig.~\ref{fig.intro}, $H_1$ can be  decomposed into two social contexts $H_4$ and $H_5$, which requires the setting of $t=3$. However, the identification of $H_2$ and $H_3$ requires $t=2$. This indicates the necessary of personalized parameter $t$ for different social contexts. 

\end{itemize}

To address the above two limitations, we define a novel notation of \discomp to represent each distinct social context without inputing any parameters. Specifically, a \discomp is a densest and maximal connected subgraph inside a user's contact neighborhood. It can be regarded as a criteria for representing unique and strong social context. However, the distribution of \discomps in two users' contact neighborhoods can be totally different in terms of density and quantity, which cannot be compared directly. To tackle this issue, we propose a new structural diversity model based on $h$-index. In the literature, the $h$-index is defined as the maximum number of $h$ such that a researcher has published $h$ papers whose citations have at least $h$~\cite{hirsch2005index}. We apply the similar idea to measure structural diversity in ego-networks. Given a vertex $v$, the structural diversity of $v$ is the largest number $h$ such that there exists at least $h$ \discomps with coreness at least $h$. In this paper, we study the problem of top-$k$ $h$-index based structural diversity search, which finds  $k$ vertices with largest  $h$-index based structural diversity. To summarize, we make the following contributions:

\begin{itemize}
       \vspace{-0.2cm}
    \item We propose a novel definition of \discomp to provide a parameter-free scheme for identifying social contexts. To simultaneously measure the quantity and strength of social contexts in one's contact neighborhood, we propose a new $h$-index based structural diversity model. We formulate the problem of top-$k$ $h$-index based structural diversity search in a graph. (Section~\ref{sec.problem})

    \item We propose a useful approach for computing the $h$-index based structural diversity score $h(v)$ for a vertex $v$ and give a baseline algorithm for solving the top-$k$ structural diversity search problem. (Section~\ref{sec.baseline})

    \item Based on the analysis of the \discomp structure and the property of $h$-index, we design an upper bound of $h(v)$. Equipped with the upper bound, we propose an efficient top-$k$ search framework to improve the efficiency. (Section~\ref{sec.efficientalg}).

    \item We conduct extensive experiments on four real-world large datasets to demonstrate the parameter sensitivity of the existing core-based structural diversity model and verify the effectiveness of our proposed model. Experiment results also validate the efficiency of our proposed algorithms. (Section~\ref{sec.exp})

\end{itemize}

\section{Related Work}\label{sec.relate}
\input{relate} 



\section{Problem Statement}\label{sec.problem}
\input{problem}


\section{Baseline Algorithm}\label{sec.baseline}
\input{baseline}
\section{Efficient Top-$k$ Search Algorithm}\label{sec.efficientalg}
\input{fastalg}

\section{Experiments}\label{sec.exp}
\input{experiment}

\section{Conclusion}\label{sec.conclusion}
In this paper, we propose a parameter-free structural diversity model based on $h$-index and study the top-$k$ structural diversity search problem. To solve the top-$k$ structural diversity search problem, an upper bound for the diversity score and a top-$k$ search framework for efficiently reducing the search space are proposed. Extensive experiments on real-wold datasets verify the efficiency of our pruning techniques and the effectiveness of our proposed $h$-index based structural diversity model.

\section*{Acknowledgments}
This work is supported by the NSFC Nos. 61702435, 61972291, RGC Nos. 12200917,  12200817, CRF C6030-18GF, and the National Science Foundation of Hubei Province No. 2018CFB519.


\bibliographystyle{ACM-Reference-Format}
\bibliography{truss}

\end{document}

%% file: relate.tex
This work is related to the studies of structural diversity search and $k$-core mining. 

\stitle{Structural Diversity Search.} In~\cite{UganderBMK12}, Ugander et al. studied the structural diversity models in the real-world applications of social contagion. The problem of top-$k$ structural diversity search is proposed and studied by Huang et al.~\cite{HuangCLQY13,HuangCLQY15}. The goal of the problem is to find $k$ vertices with the highest structural diversity scores. Two structural diversity models based on $t$-sized component and $t$-core respectively are studied w.r.t. a parameter threshold $t$. Recently, Chang et al.~\cite{DBLP:conf/icde/ChangZLQ17} proposed fast algorithms to address structural diversity search by improving the efficiency and scalability of the methods~\cite{HuangCLQY13}. Cheng et al.~\cite{cheng2019keyword} propose an approach of diversity-based keyword search to solve the mashup construction problem. Different from above studies, we propose a parameter-free structural diversity model based on the novel definition of \discomps, which avoids suffering from the difficulties of parameter tuning. 

\stitle{K-Core Mining.}
There exist lots of studies on $k$-core mining in the literature. $k$-core is a definition of cohesive subgraph, in which each vertex has degree at least $k$. The task of core decomposition is finding all non-empty $k$-cores for all possible $k$'s. Batagelj et al.~\cite{BatageljZ03} proposed an in-memory algorithm of core decomposition. Core decomposition has also been widely studied in different computing environment such as external-memory algorithms~\cite{JCheng11}, streaming algorithms~\cite{SariyuceGJWC13}, distributed algorithms~\cite{MontresorPM13}, and I/O efficient algorithms~\cite{WenQZLY19}. The study of core decomposition is also extended to different types of graphs such as dynamic graphs~\cite{JakmaOPF12,AridhiBMV16}, uncertain graphs~\cite{BonchiGKV14}, directed graphs~\cite{Levorato14}, temporal graphs~\cite{WuCLKHYW15}, and multi-layer networks~\cite{GalimbertiBG17}. Recently, core maintenance in dynamic graphs has attracted significant interest in the literature~\cite{LiYM14,AridhiBMV16,ZhangYZQ17}. 

%% file: problem.tex
In this section, we formulate the problem of $h$-index based structural diversity search. \eat{Table~\ref{tab:notations} lists the notations that will be frequently used in the paper. }

\subsection{Preliminaries}
We consider an undirected and unweighted simple graph $G=(V,E)$, where $V$ is the set of vertices and $E$ is the set of edges. We denote $n=|V|$ and $m=|E|$ as the number of vertices and edges in $G$ respectively. W.l.o.g. we assume the input graph $G$ is a connected graph, which implies that $m\geq n-1$. 
For a given vertex $v$ in a subgraph $H$ of $G$, we define $N_H(v)=\{u$ in $H : (u,v)\in E(H)\}$ as the set of neighbors of $v$ in $H$, and $d_H(v)=|N_H(v)|$ as the degree of $v$ in $H$. We drop the subscript of $N_G(v)$ and $d_G(v)$ if the context is exactly $G$ itself, i.e. $N(v)$, $d(v)$. The maximum degree of graph $G$ is denoted by $d_{max}=\max_{v\in V} d_G(v)$.

Given a subset of vertices $S\subseteq V$, the subgraph of $G$ induced by $S$ is denoted by $G_S=(S,E(S))$, where the edge set $E(S)=\{(u,v)\in E: u,v\in S \}$. \eat{We denote the degree of a vertex $v$ in a subgraph $S$ by $d_{G_S}(v)$.} Based on the definition of induced subgraph, we define the \ego~\cite{DBLP:journals/apin/DingZ18,mcauley2014discovering} 
  as follows.


\begin{definition}{(Ego-network)}
Given a vertex $v$ in graph $G$, the \ego of $v$ is the induced subgraph of $G$ by its neighbors $N(v)$, denoted by $G_{N(v)}$.
\end{definition}

In the literature, the term ``neighborhood induced subgraph''~\cite{HuangCLQY15} is also used to describe the \ego of a vertex. For example, consider the graph $G$ in Fig.~\ref{fig.intro}. The \ego of vertex $v$ is shown in the gray area of Fig.~\ref{fig.intro}, which excludes $v$ itself with its incident edges. The $t$-core of a graph $G$ is the largest subgraph of $G$ in which all the vertices have degree at least $t$. However, the $t$-core of a graph can be disconnected, which may not be suitable to directly depict social contexts. Hence, we define the connected $t$-core as follows. 

\begin{definition}{(Connected $t$-Core)}
Given a graph $G$ and a positive integer $t$, a subgraph $H\subseteq G$ is called a connected $t$-Core iff $H$ is connected and each vertex $v\in V(H)$ has degree at least $t$ in $H$.
\end{definition}



Given a parameter $t$, the core-based structural diversity model treats each maximal connected $t$-core as a distinct social context~\cite{HuangCLQY15}\cite{UganderBMK12}. To measure the structural diversity of an \ego, one essential step is to tune a proper value for parameter $t$. However, such parameter setting is not easy and even critically challenging. The following example illustrates it. 



\begin{example}
Fig.~\ref{fig.intro} shows an \ego $G_{N(v)}$ of vertex $v$. Given an integer $t=2$, three maximal connected $2$-core ($H_1$, $H_2$ and $H_3$) will be treated as distinct social contexts. The core-based structural diversity of $v$ is 3. When we set $t=3$, the core-based structural diversity of $v$ will be 2, since $H_4$ and $H_5$ will be treated as two distinct social contexts. In this case, $H_2$ and $H_3$ are no longer treated as social contexts. 
If we set $t$ to be some values higher than $3$, no social contexts can be identified. The core-based structural diversity of $v$ will then be 0. From this example, we can see that if the value of $t$ is tuned too high, no social contexts can be identified. But if the value of $t$ is set too low, some strong social contexts with denser structures cannot be captured. Thus, to choose a proper value of $t$ for all vertices in a graph is a challenging task. 
\end{example}

To tackle the above issue, we propose a parameter-free scheme for automatically identifying strong social contexts in one's \ego. We firstly give a novel definition of \discomp based on the concept of coreness as follows. 

\begin{definition}{(Coreness)}\label{def.coreness}
Given a subgraph $H\subseteq G$, the coreness of $H$ is the minimum degree of vertices in $H$, denoted by $\varphi(H)=\min_{v\in H}\{d_H(v)\}$. The coreness of a vertex $v\in V(G)$ is $\varphi_G(v)=\max_{H\subseteq G, v\in V(H)}\{\varphi(H)\}$. 
\end{definition}

\begin{definition}(Discriminative Core)
\label{def.discore}
Given a graph $G$ and a subgraph $H\subseteq G$, $H$ is a \discomp if and only if $H$ is a maximal connected subgraph such that there exists no subgraph $H' \subseteq H$ with $\varphi(H') > \varphi(H)$.
\end{definition}

By Def.~\ref{def.discore}, a \discomp $H$ is a maximal connected component that cannot be further decomposed into smaller subgraphs with a higher coreness. It indicates that a \discomp is the densest and most important component of a social context, which can be used as a distinct element  to represent a social context. 
In addition, the coreness of a \discomp  reflects the strength of its representative social context. For example, $H_4$ is a \discomp with $\varphi(H_4)=3$. And $H_2$ is another \discomp with $\varphi(H_2)=2$. According to the core-based structural diversity, they cannot be identified as distinct social contexts simultaneously using the same value of parameter $t$. But by our \discomp definition, they will be treated as distinct social contexts automatically without loosing the information of their strength.



For an \ego $G_{N(v)}$, the whole network may consist of multiple \discomps with various corenesses, which can be depicted as a coreness distribution of discriminative cores. Moreover, to rank the structural diversity of two vertices, it is difficult to directly compare the coreness distributions of two \egos. Because it is not easy to measure both the number of social contexts and the strength of social contexts simultaneously. 

Making use of the idea of $h$-index criteria, we define the diversity vector and diversity score as follows.



\begin{definition}(Diversity Vector and Diversity Score)\label{def-hindex}
Given a graph $G$ and a vertex $v$, the diversity vector of $v$ is the coreness distribution of discriminative cores in $G_{N(v)}$, denoted by $\mathcal{C}(v)=[c_v(1),...,c_v(n)]$, where $c_v(r) = |\{H: \varphi(H) =r$ and $H$ is a discriminative  core in  $G_{N(v)}\}|$. The $h$-index based structural diversity score of $v$, denoted by $h(v)$, is defined as $h(v)=\max\{r : \sum_{r}^nc_v(r)\geq r\}$. For short, diversity score is called.
\end{definition}



\begin{example}
Consider the \ego of $v$ shown in Fig.~\ref{fig.intro}, subgraph $H_1$ is not a $2$-core discriminative component since it can be further decomposed into two $3$-cores $H_4$ and $H_5$. There is no \discomp with the coreness of 1, so $c_v(1)=0$. And $c_v(2)=2$ since it has two \discomps $H_2$ and $H_3$ with the coreness of 2. Similarly, $c_v(3)=2$ because $H_4$, $H_5$ are two \discomps with the coreness of 3. There exists no \discomps with coreness greater than 3. Thus, the diversity vector of $v$ is $\mathcal{C}(v)=[0,2,2,0,...,0]$. And the diversity score is $h(v)=2$ by definition.
\end{example}

In this paper, we study the problem of $h$-index based structural diversity search in a graph. The problem formulation is defined as follows. 

\begin{flushleft}
\textbf{Problem Formulation.} Given a graph $G$ and an integer $k$, the goal of $h$-index based structural diversity search problem is to find an optimal answer $S^*$ consisted of $k$ vertices with the highest $h$-index based structural diversity scores, i.e., 
$$S^* = \mathop{\arg\max}_{S\subseteq V, |S|=k} \{\min_{v\in S} h(v) \}.$$
\end{flushleft}

%% file: baseline.tex
In this section, we introduce a baseline approach for h-index based structural diversity search over graph $G$.  The high-level idea is to compute the diversity score for each vertex in graph $G$ one by one. After obtaining the scores of all vertices, it sorts vertices in decreasing order of their scores and returns the first $k$ vertices with the highest structural diversity scores. This method computes the top-$k$ result from scratch, which is intuitive and straightforward to obtain answers.

In the following, we first introduce an existing algorithm of core decomposition~\cite{BatageljZ03}. Then, we present an important and useful procedure to compute h-index based structural diversity score $h(v)$ for a given vertex $v$. 

\subsection{Core Decomposition}
The core decomposition of graph $G$ computes the coreness of all vertices $v\in V$. Algorithm~\ref{algo:core-decomp} outlines the algorithm of core decomposition~\cite{BatageljZ03}.
The algorithm starts with an integer $t=1$, and iteratively removes the nodes with degree less than  $t$ and their incident edges. The number of $t-1$ is assigned to be the coreness of the removed vertices. Then, the degree of affected vertices needs to be updated, since the removal of a vertex  decreases the degree of its neighbors in the remaining  graph. The number $t$ is increased by one after each iteration, until all vertices and edges are deleted from the input graph. 



\subsection{Computing $h(v)$}
The computation of $h(v)$ includes three major steps. First, we extract from graph $G$ and obtain an \ego $G_{N(v)}$ for vertex $v$, which is the induced subgraph of $G$ by the set of $v$'s neighbors $N(v)$. Next, we decompose the entire  \ego $G_{N(v)}$ into several discriminative cores, and count their corenesses to derive structural diversity vector $\mathcal{C}(v)$. The detailed procedure is outlined in Algorithm~\ref{algo:disc-decomp}.  Finally, based on the diversity vector of $\mathcal{C}(v)$, we compute the diversity score $h(v)$ by the Def.~\ref{def-hindex} using Algorithm~\ref{algo:baseline-score}.


\begin{algorithm}[t]
\small
\caption{Core Decomposition~\cite{BatageljZ03}} \label{algo:core-decomp}
\begin{flushleft} 
\textbf{Input:} a graph $G=(V, E)$\\
\textbf{Output:} the coreness $\varphi_G(v)$ for each vertex $v\in V$ \\
\end{flushleft}
\begin{algorithmic}[1]

\STATE  $\mathcal{L} \leftarrow$ Sort all vertices in $G$ in ascending order of their degree.

\STATE Let $t \leftarrow 1$;

\STATE  \textbf{while} $G$ is not empty \textbf{do}

\STATE \hspace{0.3cm} \textbf{for} each vertex $v\in \mathcal{L}$ with $d(v)<t$ \textbf{do}

\STATE \hspace{0.6cm} Remove $v$ and its incident edges from $G$; Remove $v$ from $\mathcal{L}$;

\STATE \hspace{0.6cm} $\varphi_G(v) \leftarrow t-1$;

\STATE \hspace{0.6cm} Update the degree of the affected vertices and reorder $\mathcal{L}$;

\STATE \hspace{0.3cm} $t \leftarrow t+1$;

\STATE  \textbf{return} $\varphi_G(v)$ for each vertex $v\in V$; 

\end{algorithmic}
\end{algorithm}

\stitle{Discriminative Core Decomposition.} 
Algorithm~\ref{algo:disc-decomp} outlines the detailed steps for discriminative core decomposition and diversity vector computation. 
For an \ego $G_{N(v)}$ of vertex $v$, we firstly apply the core decomposition algorithm on it to calculate the coreness of each vertex (line 1). Then, we sort all vertices in $G_{N(v)}$ in ascending order of their coreness (line 3). For each integer $t$ from 1 to the maximum coreness of the vertices in $G_{N(v)}$, we identify and count the number of \discomps with the coreness of $t$ by using a breadth first search approach (lines 5-19). By definition, a \discomp with the coreness of $t$ will be only formed by the vertices with the coreness of exactly $t$. Thus, in each iteration, we traverse vertices with the same coreness of $t$ to search all the \discomps $H$s with $\varphi(H)=t$ (lines 7-19 and lines 14-15). Edges connecting the current visited vertex $x$ to the vertices with coreness greater than $t$ indicate that the current found component can not be counted as a \discomp and $x$ does not belong to any \discomps in $G_{N(v)}$ (lines 16-17). Then the $t$-th element $c_v(t)$ of the diversity vector $\mathcal{C}(v)$ can be computed (lines 18-19). Finally, the diversity vector $\mathcal{C}(v)$ of $v$ will be returned.

\begin{algorithm}[t]
\small
\caption{Discriminative Core Decomposition} \label{algo:disc-decomp}
\begin{flushleft} 
\textbf{Input:} an \ego $G_{N(v)}=(N(v),  \{(u,w)\in E: u,w\in N(v)\}))$\\
\textbf{Output:} the diversity vector $\mathcal{C}(v)$\\
\end{flushleft}
\begin{algorithmic}[1]


\STATE 	Apply the core decomposition algorithm in Algorithm ~\ref{algo:core-decomp} on $G_{N(v)}$;

\STATE 	$t_{max}= \max_{u\in N(v)}\varphi_{G_{N(v)}}(u)$;

\STATE $\mathcal{L} \leftarrow$ Sort all vertices in $G_{N(v)}$ in ascending order of their coreness;

\STATE $Q \leftarrow \emptyset$; $visited \leftarrow \emptyset$

\STATE \textbf{for} $t \leftarrow 1$ to $t_{max}$ \textbf{do}

\STATE \hspace{0.3cm} $c_{v}(t) \leftarrow 0$;

\STATE \hspace{0.6cm} \textbf{for} each vertex $u\in \mathcal{L}$ with the coreness of $\varphi_{G_{N(v)}}(u)=t$ \textbf{do}

\STATE \hspace{0.6cm} \hspace{0.3cm} $Flag \leftarrow$ \textbf{true};

\STATE \hspace{0.6cm} \hspace{0.3cm} \textbf{if} $u\notin visited$ \textbf{then};

\STATE \hspace{0.6cm} \hspace{0.3cm} \hspace{0.3cm} $visited \leftarrow visited\cup \{u\}$; $Q.push(u)$;

\STATE \hspace{0.6cm} \hspace{0.3cm} \hspace{0.3cm} \textbf{while} $Q$ is not empty \textbf{do}

\STATE \hspace{0.6cm} \hspace{0.3cm} \hspace{0.3cm} \hspace{0.3cm} $x \leftarrow Q.pop()$;

\STATE \hspace{0.6cm} \hspace{0.3cm} \hspace{0.3cm} \hspace{0.3cm} \textbf{for} each $y\in \{y : (x,y)\in E(G_{N(v)})\}$ \textbf{do}

\STATE \hspace{0.6cm} \hspace{0.3cm} \hspace{0.3cm} \hspace{0.3cm} \hspace{0.3cm} \textbf{if} $\varphi_{G_{N(v)}}(y)=t$ \textbf{then}

\STATE \hspace{0.6cm} \hspace{0.3cm} \hspace{0.3cm} \hspace{0.3cm} \hspace{0.3cm} \hspace{0.3cm} Insert $y$ to $Q$ and $visited$ if $y$ is unvisited;

\STATE \hspace{0.6cm} \hspace{0.3cm} \hspace{0.3cm} \hspace{0.3cm} \hspace{0.3cm} \textbf{else if} $\varphi_{G_{N(v)}}(y)>t$ \textbf{then}

\STATE \hspace{0.6cm} \hspace{0.3cm} \hspace{0.3cm} \hspace{0.3cm} \hspace{0.3cm} \hspace{0.3cm} $Flag \leftarrow$ \textbf{false};

\STATE \hspace{0.6cm} \hspace{0.3cm} \hspace{0.3cm} \textbf{if} $Flag =$ \textbf{true} \textbf{then};

\STATE \hspace{0.6cm} \hspace{0.3cm} \hspace{0.3cm} \hspace{0.3cm} $c_{v}(t) \leftarrow c_{v}(t)+1$;

\STATE \textbf{return} $\mathcal{C}(v)$;















\end{algorithmic}
\end{algorithm}

\begin{algorithm}[t]
\small
\caption{Compute $h(v)$} \label{algo:baseline-score}
\begin{flushleft} 
\textbf{Input:} a graph $G=(V, E)$; a vertex $v$\\
\textbf{Output:} the diversity score $h(v)$\\
\end{flushleft}
\begin{algorithmic}[1]

\STATE  Extract the \ego $G_{N(v)}$ of $v$;

\STATE 	$\mathcal{C}(v) \leftarrow$ Apply the discriminative core decomposition procedure in Algorithm~\ref{algo:disc-decomp} on $G_{N(v)}$;

\STATE  $h(v) \leftarrow 0$;

\STATE  \textbf{for} $t \leftarrow t_{max}$ to $1$ \textbf{do}

\STATE  \hspace{0.3cm}	$h(v) \leftarrow h(v)+ c_{v}(t)$

\STATE  \hspace{0.3cm}	\textbf{if} $h(v) \geq t$ \textbf{then} $h(v)\leftarrow t$; \textbf{break};

\STATE  \textbf{return}  $h(v)$;

\end{algorithmic}
\end{algorithm}

\stitle{H-index Score Computation.} The details of computing the $h$-index based structural diversity score are shown in Algorithm~\ref{algo:baseline-score}. After figuring out the diversity vector $\mathcal{C}(v)$ (lines 1-2), the  diversity score $h(v)$ can then be calculated by Def.~\ref{def-hindex} (lines 3-6). We firstly initialize $h(v)$ as 0 (line 3). Then, for each element $c_v(t)$ in the reverse order of the diversity vector $\mathcal{C}(v)$, we keep accumulating it to $h(v)$ until the first $t$ appears such that $h(v)\geq t$ (line 4-6). Such $t$ is the diversity score $h(v)$ of $v$.

Equipped with Algorithm~\ref{algo:baseline-score}, we are able to compute the $h$-index based structural diversity for all the vertices in $G$. By sorting the diversity scores, we can obtain the top-$k$ results for a given $k$. 

%% file: fastalg.tex
The drawback of baseline method presented in the previous section is obviously inefficient and can be improved. \eat{First, it iteratively computes the $h$-index based structural diversity scores for all vertices on the entire graph $G$, which is expensive. Some vertices unqualified for the top-$k$ answer should be avoided for score computations. Moreover, the processes of \ego extraction and discriminative core decomposition are also costly in computation.}\jinbin{Firstly, both the \ego extraction and discriminative core decomposition are costly in computation. Secondly, it iteratively computes the $h$-index based structural diversity scores for all vertices on the entire graph $G$, which is expensive. Thirdly, some vertices appear to be obviously unqualified for the top-$k$ result. And the score computations of them are reluctant and should be avoided.} 

In this section, we develop an efficient top-$k$ search framework by exploiting useful pruning techniques to reduce the search space, leading to a small number of candidate vertices for score computations. Specifically, we design an upper bound $\widehat{h}(v)$ for diversity score $h(v)$, based on the analysis of the core structure. 

\subsection{An Upper Bound of $h(v)$}
We starts with a structural property of $t$-core.

\begin{lemma}\label{lemma.globalcore}
Given a vertex $v$ and any vertex $u \in N(v)$, if $u$ has $\varphi_{G_{N(v)}}(u)=r$ in \ego $G_{N(v)}$, then $u$ has the coreness $\varphi_G(u)\geq r+1$ in graph $G$.
\end{lemma}

\begin{proof}
We omit the proof for brevity. The detailed proof can be referred to~\cite{HuangCLQY15}.
\end{proof}

\begin{example}
Consider vertex $x_1$ in Fig.~\ref{fig.intro}, $x_1$ has coreness $\varphi_G(x_1)=4$. However, in the \ego $G_{N(v)}$, $\varphi_{G_{N(v)}}(x_1)=3$. Here $\varphi_G(x_1)\geq \varphi_{G_{N(v)}}(x_1)+1$ holds.
\end{example}

For a vertex $v$ and some vertices $u\in N(v)$, the global coreness $\varphi_G(u)$ is sometimes much larger than the coreness of $u$ in the \ego of $v$, i.e. $\varphi_G(u) >> \varphi_{G_{N(v)}}(u)$. The following lemma gives another upper bound for estimating the coreness $\varphi_{G_{N(v)}}(u)$, w.r.t. vertices $v$ and $u\in N(v)$.

\begin{lemma}\label{lemma.localcore}
Given a vertex $v$ and its coreness $\varphi_G(v)$, $\forall u\in N(v)$, $\varphi_{G_{N(v)}}(u) < \varphi_G(v)$.
\end{lemma}

\begin{proof}
We prove this by contradiction. For any $u\in N(v)$, we assume $\varphi_G(v)=r$ and $\varphi_{G_{N(v)}}(u)\geq \varphi_G(v)$, which is $\varphi_{G_{N(v)}}(u)\geq r$. By the definition of coreness, there exists a subgraph $H\subseteq G_{N(v)}$ with coreness $\varphi(H)\geq r$ indicating that $\forall v^*\in V(H)$, $d_H(v^*)\geq r$. We add the vertex $v$ and its incident edges to $H$ to generate a new subgraph $H'\subseteq G$, where $V(H')=V(H)\cup \{v\}$ and $E(H')=E(H)\cup \{(v,u) : u\in V(H)\}$. It's easy to verify that for all $v^*$ in $H'$, we have $d_{H'}(v^*)\geq r+1$. Since $v$ is also contained in $H'$, by definition, $\varphi_G(v)\geq r+1$, which contradicts to the condition $\varphi_G(v)=r$.  
\end{proof}


Combining Lemma~\ref{lemma.globalcore} and Lemma~\ref{lemma.localcore}, we have the following corollary.

\begin{corollary}\label{coro.bound}
Given a vertex $v$ in graph $G$, for any vertex $u \in N(v)$, $\widehat{\varphi}_{G_{N(v)}}(u)= \min\{\varphi_G(v), \varphi_G(u)-1\}$ and $\widehat{\varphi}_{G_{N(v)}}(u) \geq \varphi_{G_{N(v)}}(u)$ hold.
\end{corollary}

Based on Corollary~\ref{coro.bound}, we derive an upper bound $\widehat{h}(v)$ for the $h$-index based structural diversity score $h(v)$ as follows.


\begin{lemma}\label{lemma.bound}
Given a vertex $v$ and its \ego $G_{N(v)}$, we have an upper bound of diversity score $h(v)$, denoted by $$\widehat{h}(v)=\max_{x \in \mathbb Z_+}\{x : |\{u\in N(v): \widehat{\varphi}_{G_{N(v)}}(u)\geq x\}|\geq x\cdot(x+1)\}.$$
\end{lemma}

\begin{proof}
Assume that $h(v) = x^*$, we prove $\widehat{h}(v)\geq x^*$.  By $h(v) =x^*$, it indicates that there exists $x^*$ \discomps $g$ with $\varphi(g)\geq x^*$ in the \ego $G_{N(v)}$. For $\varphi(g)\geq x^*$, \discomp $g$ has at least $x^*+1$ nodes $u$ with $\varphi_{G_{N(v)}}(u)\geq x^*$. Thus, the whole \ego $G_{N(v)}$ has at least $x^*\cdot(x^*+1)$ nodes $u$ with $\varphi_{G_{N(v)}}(u)\geq x^*$, i.e., $h(v) =x^* \leq \max_{x \in \mathbb Z_+}\{x : |\{u\in N(v): \varphi_{G_{N(v)}}(u)\geq x\}|\geq x\cdot(x+1)\}$. By Corollary~\ref{coro.bound}, $\widehat{\varphi}_{G_{N(v)}}(u)\geq \varphi_{G_{N(v)}}(u)$, hence we have $\widehat{h}(v)\geq x^* = h(v)$.
\end{proof}

According to Lemma~\ref{lemma.bound}, once applying the core decomposition algorithm on graph $G$, we can directly compute the upper bounds $\widehat{h}(v)$ for all vertices $v$.





\subsection{Top-K Structural Diversity Search Framework}
Equipped with the upper bound $\widehat{h}(v)$, we develop an efficient top-$k$ search framework for safely pruning the search space and avoiding the unnecessary computation of $h(v)$. The efficient top-k structural diversity search framework is presented in Algorithm~\ref{algo:bound-search}.

Algorithm~\ref{algo:bound-search} starts with the initialization of the upper bound of each vertex $v$ (lines 1-2). Then, it sorts all vertices in descending order according to their upper bounds (line 3). It maintains a list $\mathcal{S}$ to store the top-$k$ result (line 4). In each iteration, the algorithm pops out a vertex $v^*$ from the vertex list $\mathcal{L}$ with the largest upper bound $\widehat{h}(v^*)$ (line 6). Next, it checks the early stop condition: if the answer set $\mathcal{S}$ has $k$ results and the minimum score in $\mathcal{S}$ is no less than the current upper bound, i.e. $\widehat{h}(v^*)\leq min_{v\in \mathcal{S}}h(v)$, the current vertex $v^*$ is safely pruned and the searching process is terminated (lines 8-9). Otherwise, the procedure of structural diversity score computation is invoked and check if $v^*$ can be added into the result set (lines 10-14). Finally, the top-$k$ results stored in $\mathcal{S}$ are returned.

\begin{algorithm}[t]
\small
\caption{Efficient Top-$k$ Search Framework} \label{algo:bound-search}
\begin{flushleft} 
\textbf{Input:} $G=(V, E)$, an integer $k$\\
\textbf{Output:} top-$k$ structural diversity results\\
\end{flushleft}
\begin{algorithmic}[1]

\STATE Apply the core decomposition on $G$ by Algorithm~\ref{algo:core-decomp} and obtain  $\varphi_G(v)$ for all vertices $v \in V$;

\STATE  \textbf{for} $v\in V$ \textbf{do}

\STATE  \hspace{0,3cm} Compute $\widehat{h}(v)$ according to Lemma~\ref{lemma.bound}; 

\STATE $\mathcal{L} \leftarrow$ Sort all vertices $V$ in descending order of $\widehat{h}(v)$;

\STATE  $\mathcal{S} \leftarrow \emptyset$;

\STATE  \textbf{while} $\mathcal{L}\neq \emptyset$ \textbf{do}

\STATE  \hspace{0.3cm} $v^* \leftarrow \arg\max_{v\in \mathcal{L}} \widehat{h}(v)$; Delete $v^*$ from $\mathcal{L}$;

\STATE  \hspace{0.3cm}  \textbf{if} $|\mathcal{S}|=r$ and $\widehat{h}(v^*) \leq \min_{v\in \mathcal{S}}h(v)$ \textbf{then}

\STATE  \hspace{0.3cm}  \hspace{0.3cm} \textbf{break};


\STATE  \hspace{0.3cm} Invoke Algorithm~\ref{algo:baseline-score}  to compute $h(v^*)$;

\STATE  \hspace{0.3cm}  \textbf{if} $|\mathcal{S}|<r$ \textbf{then} $\mathcal{S} \leftarrow \mathcal{S} \cup \{v^*\}$;

\STATE  \hspace{0.3cm}  \textbf{else if} $h(v^*) > \min_{v\in \mathcal{S}}h(v)$ \textbf{then}

\STATE  \hspace{0.3cm}  \hspace{0.3cm}  $u\leftarrow \arg\min_{v\in \mathcal{S}}h(v)$;

\STATE  \hspace{0.3cm}  \hspace{0.3cm}  $\mathcal{S} \leftarrow (\mathcal{S}-  \{u\} ) \cup \{v^*\}$;

\STATE  \textbf{return} $\mathcal{S}$; 
\end{algorithmic}
\end{algorithm}


\subsection{Complexity Analysis}

In this section, we analyze the time and space complexity of Algorithm~\ref{algo:bound-search}.

\begin{lemma}\label{lemma.basecomplex}
Algorithm~\ref{algo:baseline-score} computes $h(v)$ for each vertex $v$ in $O(\sum_{u\in N(v)}min\{d(u),d(v)\})$ time and $O(m)$ space.
\end{lemma}

\begin{proof}
Extracting $G_{N(v)}$ of $v$ takes $O(\sum_{u\in N(v)}min\{d(u),d(v)\})$, since all triangles $\triangle_{vuw}$ should be listed to enumerate each edge $(u,w) \in E(G_{N(v)})$. According to~\cite{BatageljZ03}, the core decomposition performed in $G_{N(v)}$ takes $O(|E(G_{N(v)})|+d(v))$ time. The sorting of the vertices can be finished in $O(d(v))$ time using bin sort. And the breadth first search process for identifying the \discomps needs $O(|E(G_{N(v)})|)$ time. In addition, the computing of the $h$-index based structural diversity score $h(v)$ runs in $O(\delta(G_{N(v)}))$ time, where $\delta(G_{N(v)})=\max_{u\in N(v)}\varphi_{G_{N(v)}}(u)$ is the degeneracy of $G_{N(v)}$. And $\delta(G_{N(v)})$ is bounded by the degree of $v$, which is $O(\delta(G_{N(v)}))\subseteq d(v)$. Overall, the time complexity of Algorithm~\ref{algo:baseline-score} is $O(\sum_{u\in N(v)}$ $min\{d(u),d(v)\})$.

We continue to analyze the space complexity of Algorithm~\ref{algo:baseline-score}. The storage of the \ego of $v$ takes $O(n+m)$ space since $G_{N(v)}\subseteq G$. And both the sorted list of vertices (line 4) and the structural diversity vector of $v$ takes $O(n)$ space. Thus, the space complexity of Algorithm~\ref{algo:baseline-score} is $O(n+m)\subseteq O(m)$ due to our graph connectivity assumption.
\end{proof}

\begin{theorem}
Algorithm~\ref{algo:bound-search} computes the top-$k$ results in $O(\rho m)$ time and $O(m)$ space, where $\rho$ is the arboricity of $G$ and $\rho\leq \min\{d_{max}, \sqrt{m}\}$~\cite{ChibaN85}.
\end{theorem}

\begin{proof}

Firstly, the core decomposition algorithm performed on $G$ takes $O(m)$ time and $O(n+m)$ space. Secondly, the computation of upper bound $\widehat{h}(v)$ for all $v$'s takes $O(m)$ time and $O(n) space$. In the worst case, Algorithm~\ref{algo:bound-search} needs to compute $h(v)$ for every vertex $v$. This takes $O(\sum_{v\in V}\{\sum_{u\in N(v)}min\{d(u),d(v)\}\})$ time in total by Lemma~\ref{lemma.basecomplex}. According to~\cite{ChibaN85}, we have $$O(\sum_{v\in V}\{\sum_{u\in N(v)}min\{d(u),d(v)\}\})\subseteq O(\sum_{(u,v)\in E}min\{d(u),d(v)\}) \subseteq O(\rho m).$$ Here $\rho$ is the arboricity of graph $G$, which is defined as the minimum number of disjoint spanning forests that cover all the edges in $G$. In addition, the top-$k$ results can be maintained in a list in $O(n)$ time and $O(n)$ space using bin sort. Overall, Algorithm~\ref{algo:bound-search} runs in $O(\rho m)$ time and $O(m)$ space.
\end{proof}

%% file: experiment.tex
We conduct extensive experiments on real-world datasets to evaluate the effectiveness and efficiency of our proposed $h$-index based structural diversity model and algorithms. 

\stitle{Datasets:} We run our experiments on four real-world datasets downloaded on the SNAP website~\cite{snapnets}. All datasets are treated as undirected graphs. The statistics of the networks are listed in Table~\ref{tab:dataset}. We report the node size $|V|$, edge size $|E|$ and the maximum degree $d_{max}$ of each network.

\begin{table}[t!]
\begin{center}
\scriptsize
\caption[]{\textbf{Network Statistics}}\label{tab:dataset}
\begin{tabular}{|c|c|c|c|c|c|c|c|}
\hline
Name & $|V|$ & $|E|$ & $d_{max}$ \\
\hline \hline
Gowalla	& 196,591 & 950,327 & 14,730  \\  \hline
Youtube & 1,134,890 & 2,987,624 &  28,754  \\ \hline
LiveJournal	& 3,997,962 & 34,681,189 & 14,815 \\ \hline
Orkut	& 3,072,441 & 117,185,083  & 33,313 \\ \hline

\end{tabular}
\end{center}
\vspace{-0.4cm}
\end{table}

\stitle{Compared Methods: }We evaluate all compared methods in terms of efficiency, effectiveness and also sensitivity to parameter setting. Specifically, we show three compared algorithms as follows.


\squishlisttight
\item \baseline: is the baseline method proposed in Section~\ref{sec.baseline}.
\item \hcore: is an improved top-$k$ search algorithm for computing the top-$k$ vertices with highest $h$-index based structural diversity in Algorithm~\ref{algo:bound-search}.

\item \tcore: is to compute the top-$k$ vertices with highest $t$-core based structural diversity~\cite{HuangCLQY15}. Here, $t$ is a parameter of coreness threshold.
\end{list} 



Note that in the sensitivity evaluation, we test the state-of-the-art competitor \tcore and compare the top-$k$ results for different parameter $t$. Our $h$-index based structural diversity model has no input parameter, which is consistent on the top-$k$ results.



\begin{figure*}[t]
\centering \mbox{
\subfigure[Gowalla]{\includegraphics[width=0.25\linewidth]{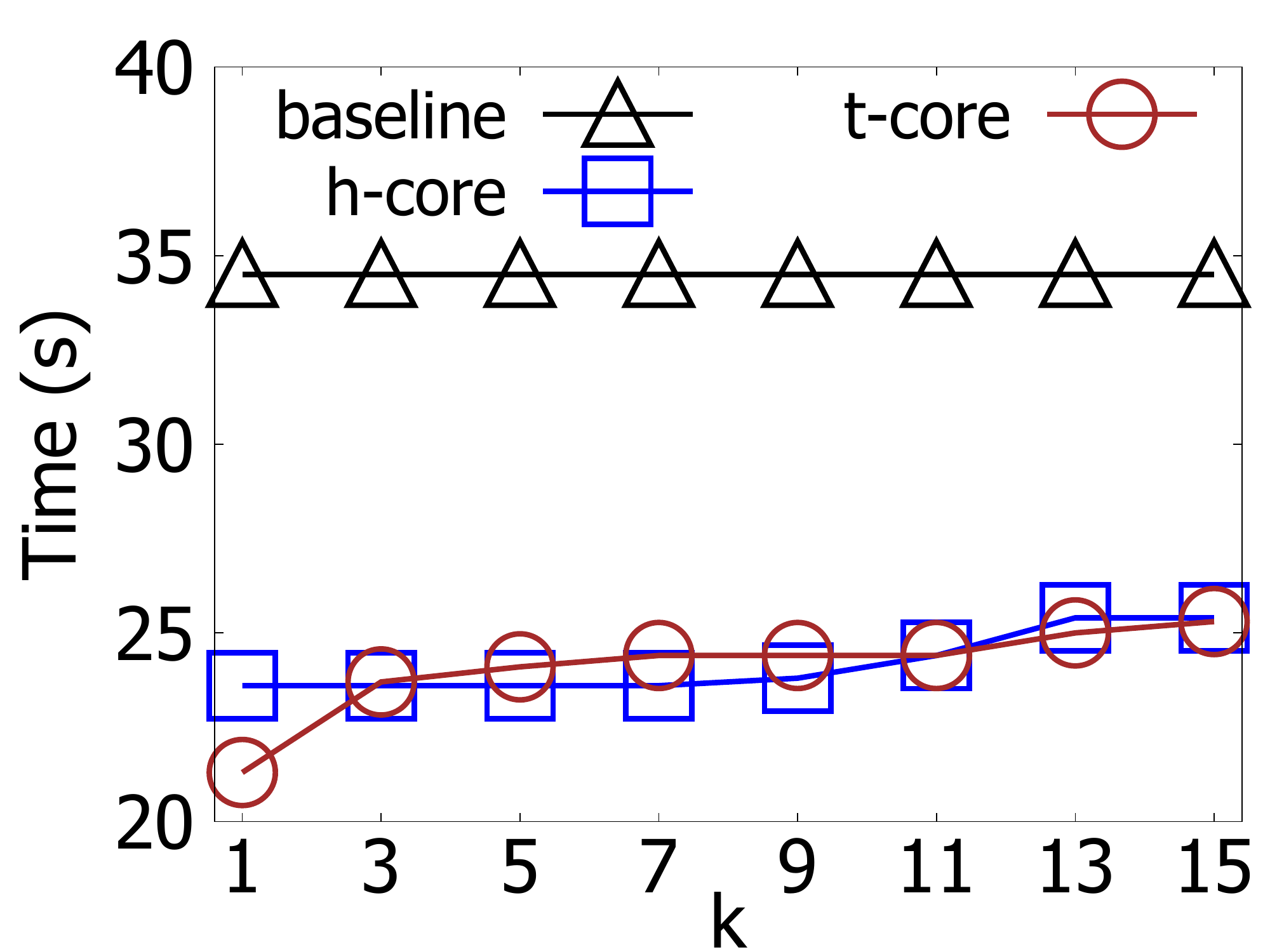}} 
\subfigure[Youtube]{\includegraphics[width=0.25\linewidth]{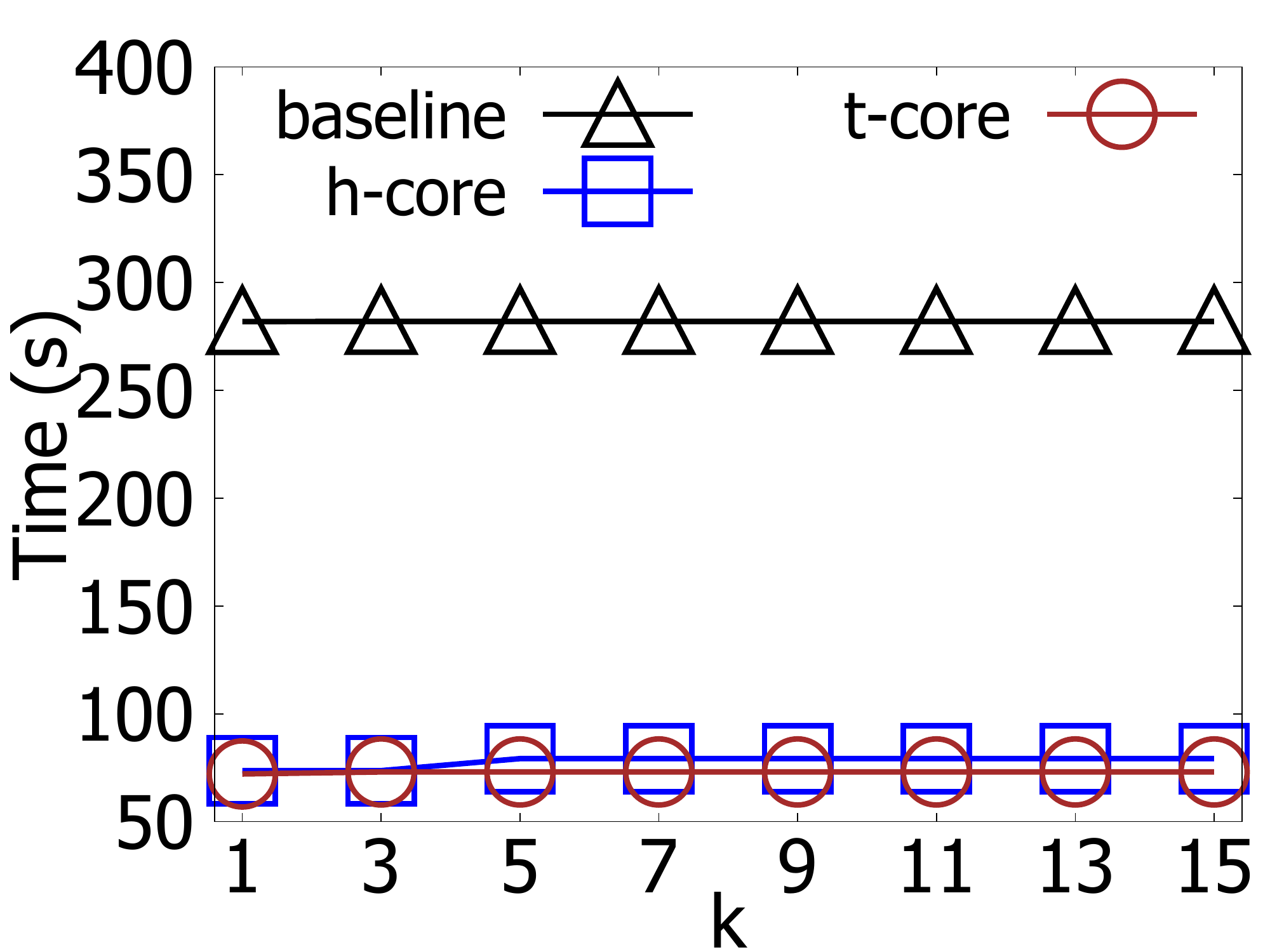}}
\subfigure[LiveJournal]{\includegraphics[width=0.25\linewidth]{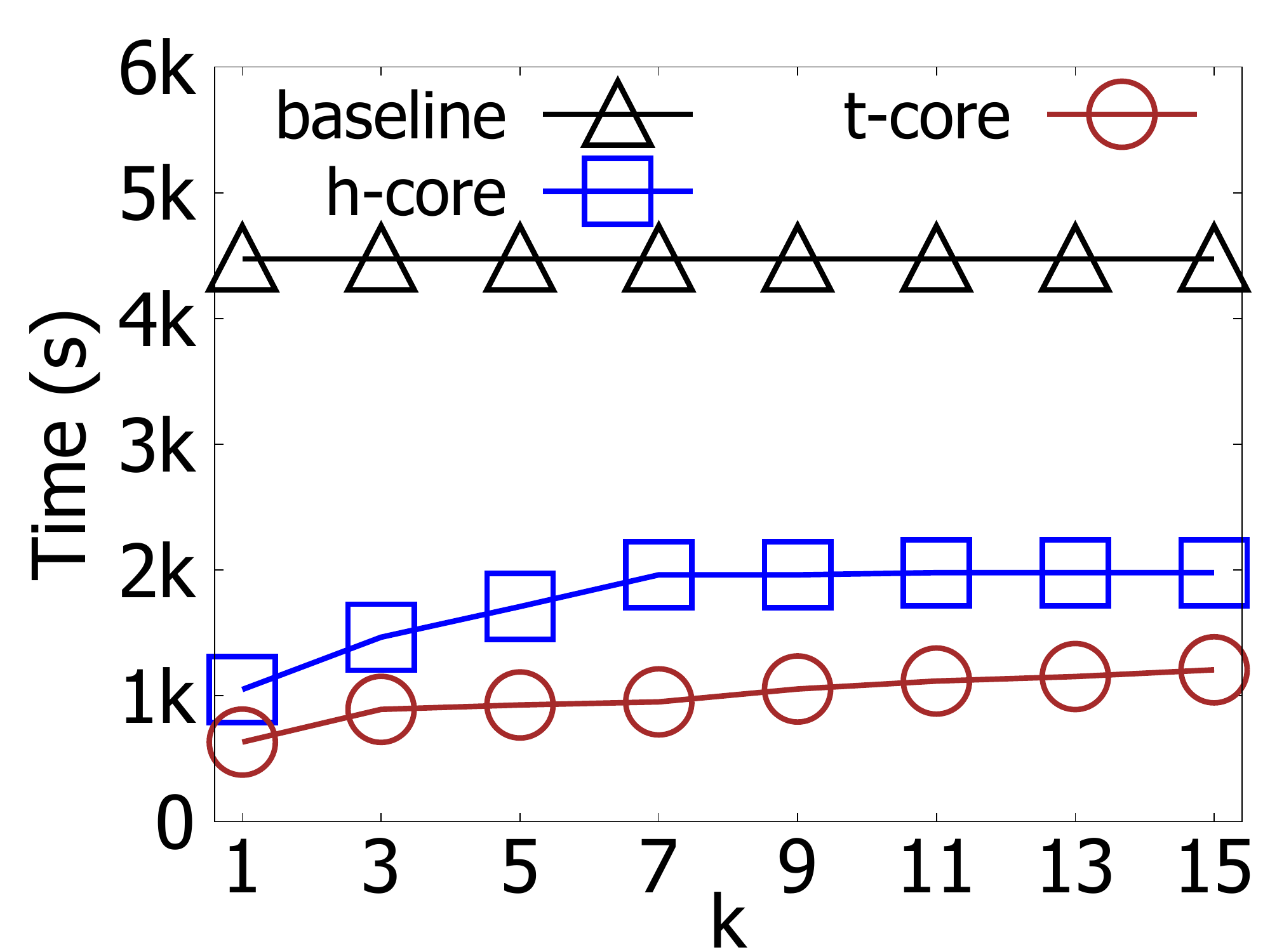}}
\subfigure[Orkut]{\includegraphics[width=0.25\linewidth]{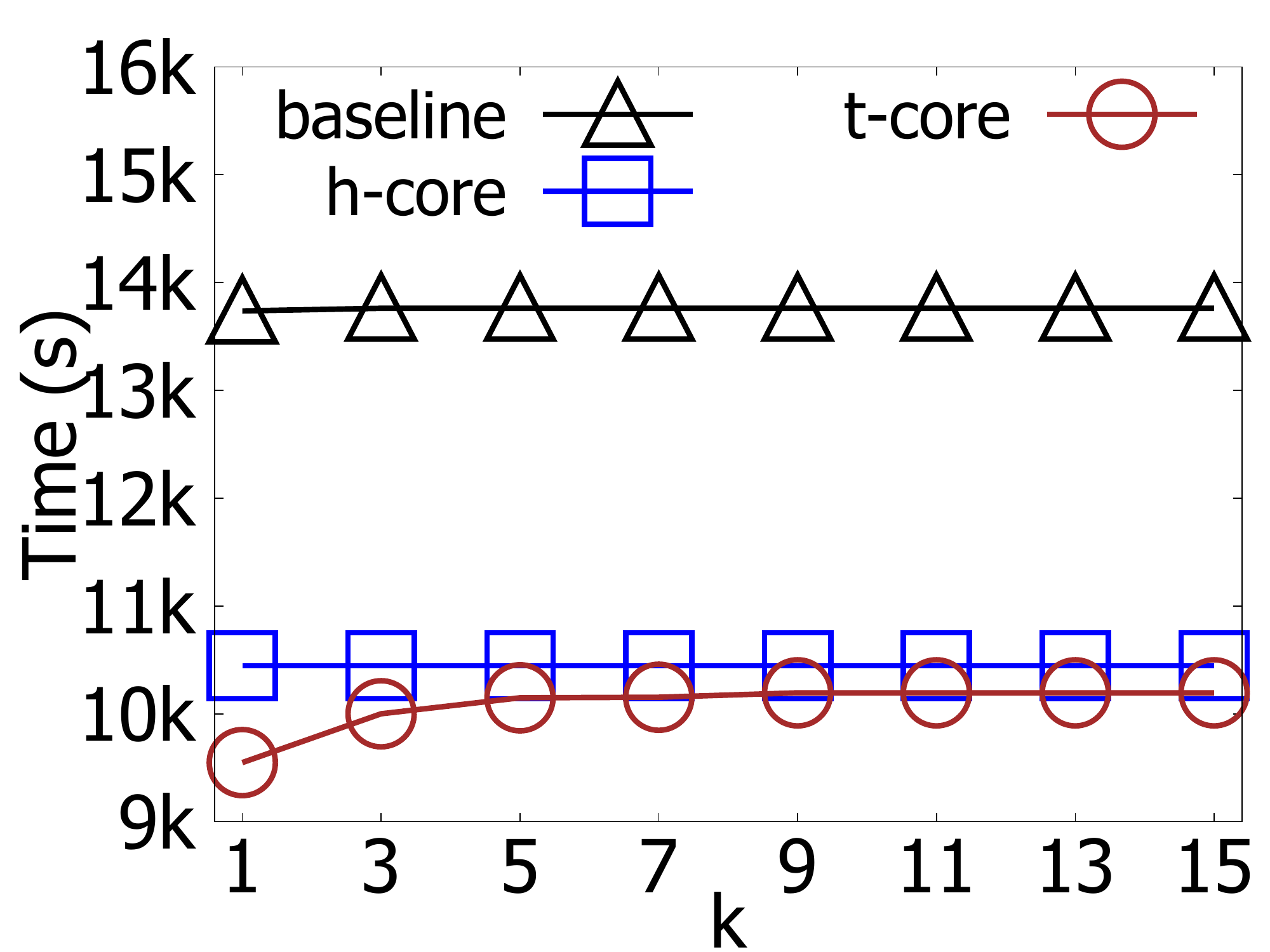}}
 }
 \vspace{-0.6cm}
\caption{Comparsion of \baseline, \bound and \tcore in terms of running time (in seconds).}
\label{fig.qt_com}
\vspace*{-0.6cm}
\end{figure*}

\begin{figure*}[t]
\centering \mbox{
\subfigure[Gowalla]{\includegraphics[width=0.25\linewidth]{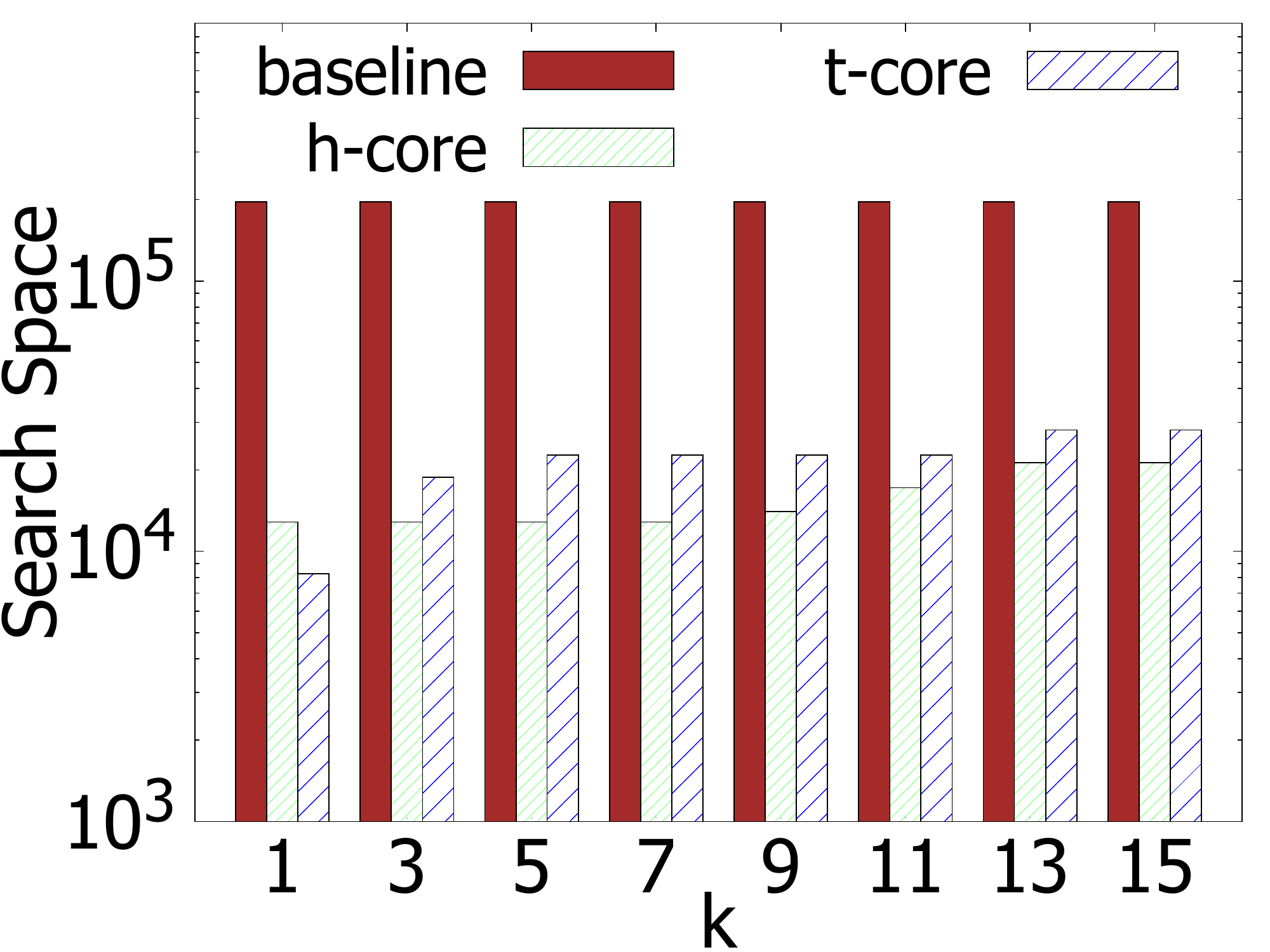}}
\subfigure[Youtube]{\includegraphics[width=0.25\linewidth]{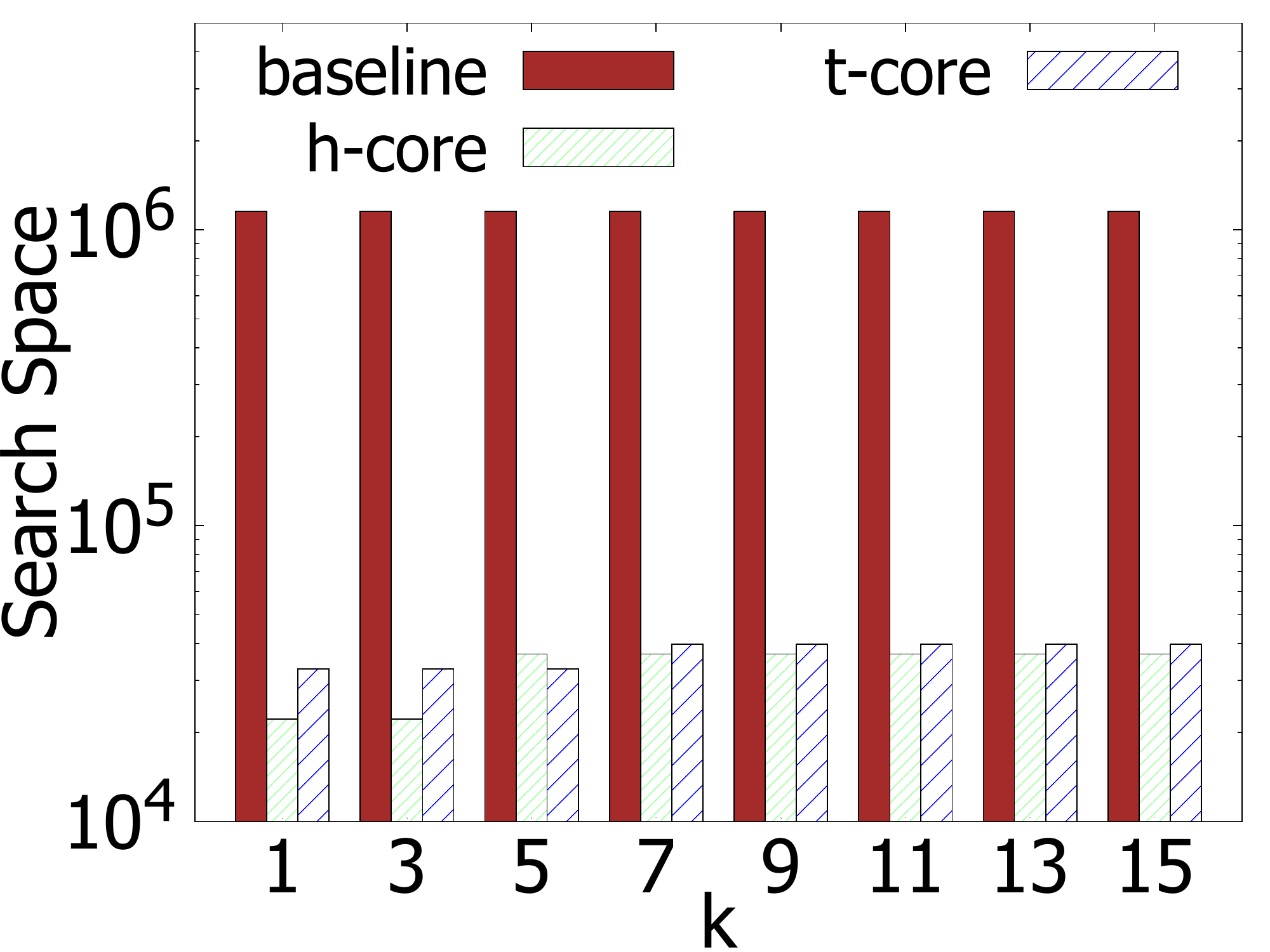}}
\subfigure[LiveJournal]{\includegraphics[width=0.25\linewidth]{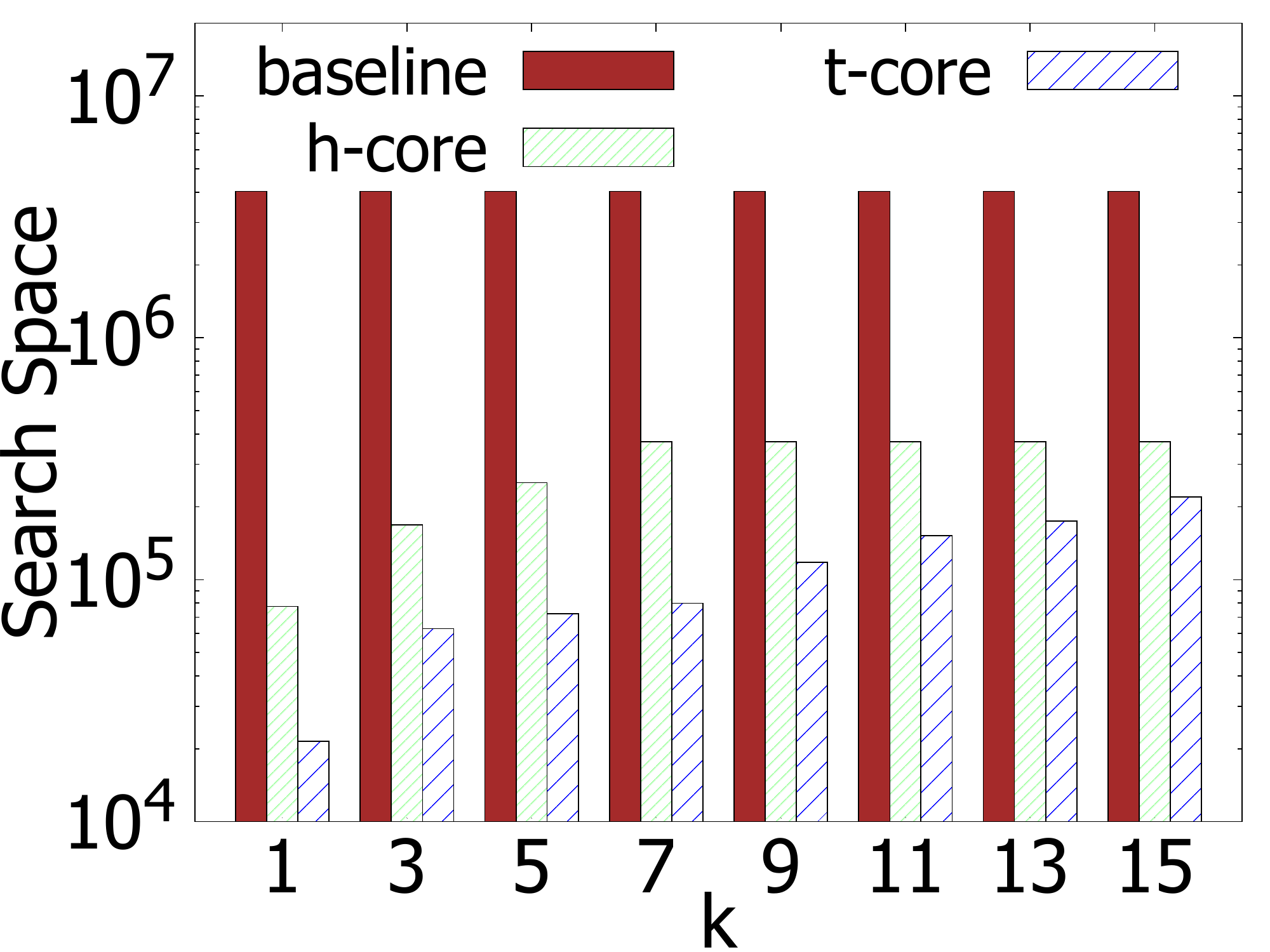}}
\subfigure[Orkut]{\includegraphics[width=0.25\linewidth]{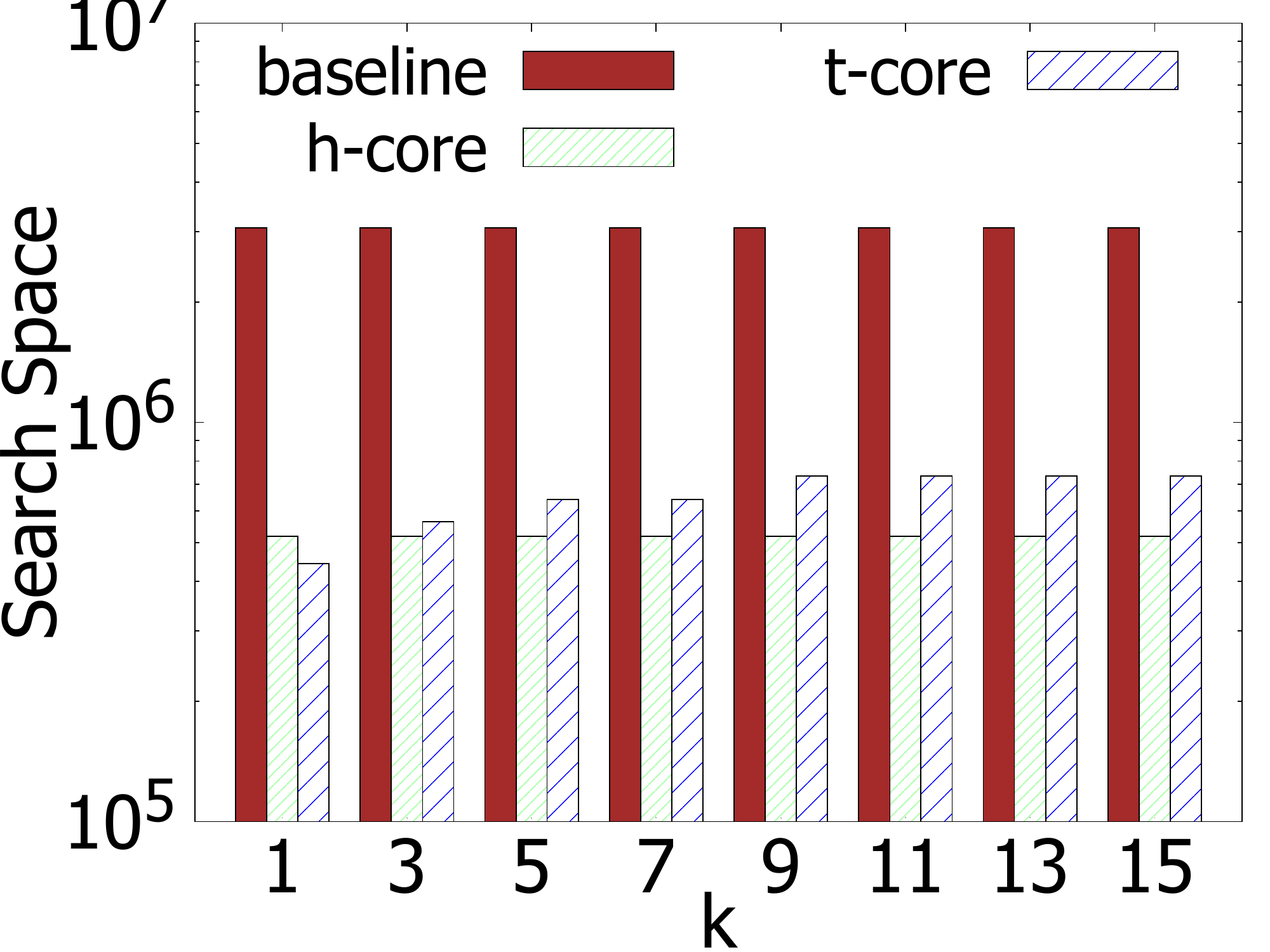}} 
 }
 \vspace{-0.6cm}
\caption{Comparsion of \baseline, \bound and \tcore in terms of search space.}
\vspace{-0.2cm}
\label{fig.sp_com}
\end{figure*}

\subsection{Efficiency Evaluation}


In this experiment, we compare the efficiency of \baseline, \bound and \tcore on four real-world datasets. For the \tcore method, we fix parameter $t=2$. We compare the running time and search space (i.e., the number of vertices whose structural diversity scores are computed in the search process). Fig.~\ref{fig.qt_com} shows the running time results of three methods varied by $k$.
It clearly shows that top-$k$ search algorithm \bound runs much faster than \baseline on all the reported datasets. Specifically, in Fig.~\ref{fig.qt_com}(c), \bound is 5 times faster than \baseline on Youtube in term of running time. Moreover, Fig.~\ref{fig.sp_com} further shows the search space of three methods varied on all datasets. We can observe that leveraging on the upper bound $\widehat{h}(v)$, a large number of disqualified vertices is pruned during the search process by \bound. The search space significantly shrinks into less than $\frac{1}{10}$ of vertex size in graphs. It verifies the tightness of our upper bound and the superiority of \bound against \baseline in efficiency. According to Fig.~\ref{fig.qt_com} and Fig.~\ref{fig.sp_com}, our \hcore is very comparative to the state-of-art method \tcore in terms of running time and search space.

\subsection{Sensitivity Evaluation}
This experiment evaluates the sensitivity of \tcore model. Given two different values of $t$, \tcore model may generate two different lists of top-$k$ ranking results. We use the Kendall rank tau distance to counts the number of pairwise disagreements between two  top-$k$ lists. The larger the distance, the more dissimilar the two lists, and also more sensitive the $t$-core model.  We adopt the Kendall distance with penalty, denoted by, 
$$K^{(p)}(\tau_1,\tau_2)=\sum_{\{i,j\}\in \mathcal{P}} \overline{K}_{i,j}^{(p)}(\tau_1,\tau_2) \nonumber$$
where $\mathcal{P}$ is the set of all unordered pairs of distinct elements in two top-$k$ list $\tau_1,\tau_2$ and $p$ is the penalty parameter. In our setting, we set $p=1$ and normalize the Kendall distance by the number of permutation $|\mathcal{P}|$. The values of normalized Kendall distance  range from 0 to 1.

We test the sensitivity of \tcore model by varying parameter $t$ in $\{2,4,6,8,10\}$. We compute the Kendall distance of two top-100 lists by \tcore model with two different $t$. The results of sensitivity heat matrix on four datasets are shown in Fig.~\ref{fig.kendalltau}. The darker colors reveal larger Kendall distances between two top-$k$ lists and also more sensitive of \tcore models on this pair of parameters $t$. Overall, sensitivity heat matrices are depicted in dark for most parameter settings on all datasets.  
This reflects that the top-$k$ results computed by \tcore are very sensitive to the setting of parameter $t$, which has a bad robustness. It strongly indicates the necessity and importance of our parameter-free structural diversity  model.

\begin{figure*}[t]
\centering \mbox{
\subfigure[Gowalla]{\includegraphics[width=0.25\linewidth]{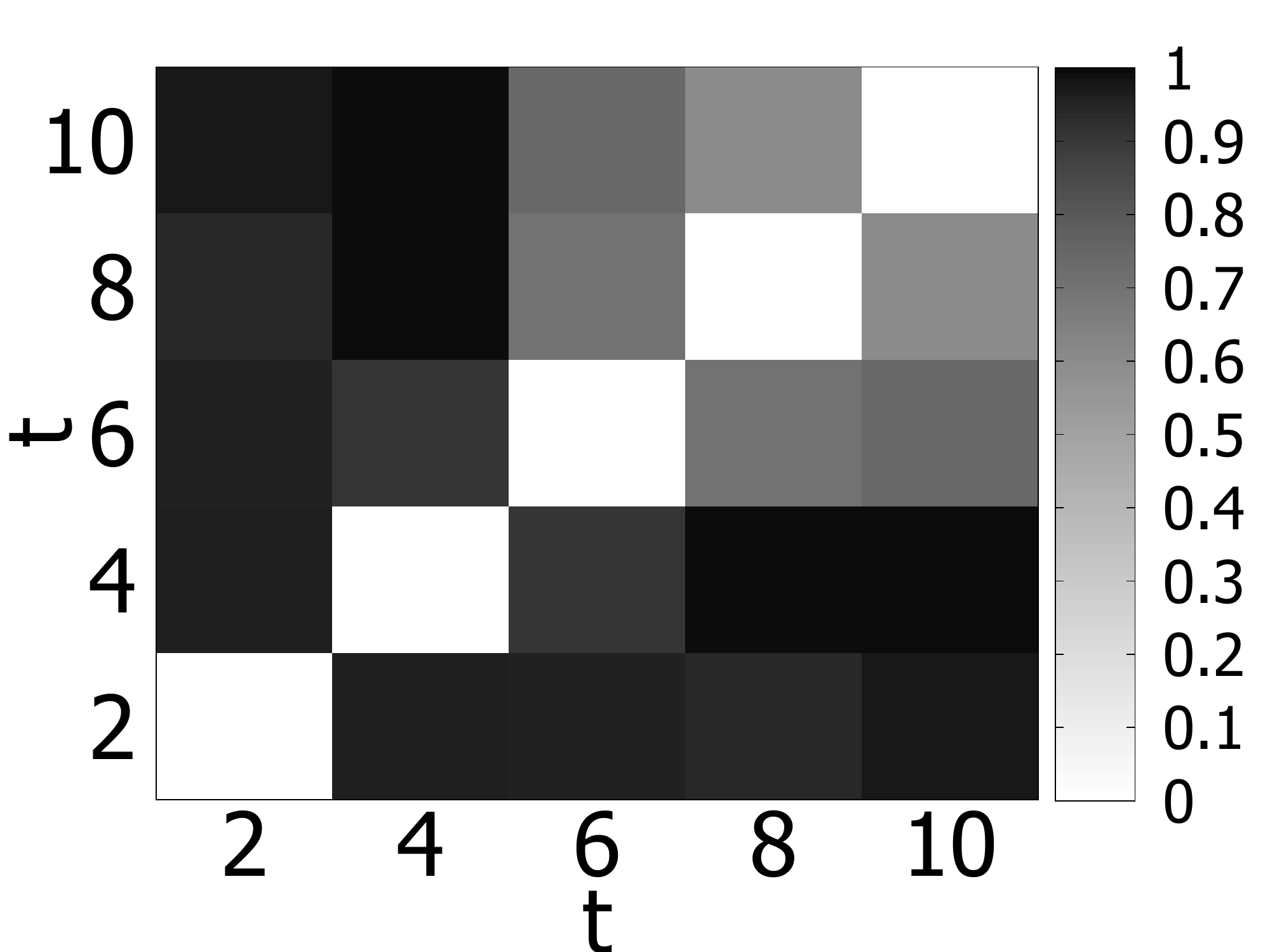}}
\subfigure[Youtube]{\includegraphics[width=0.25\linewidth]{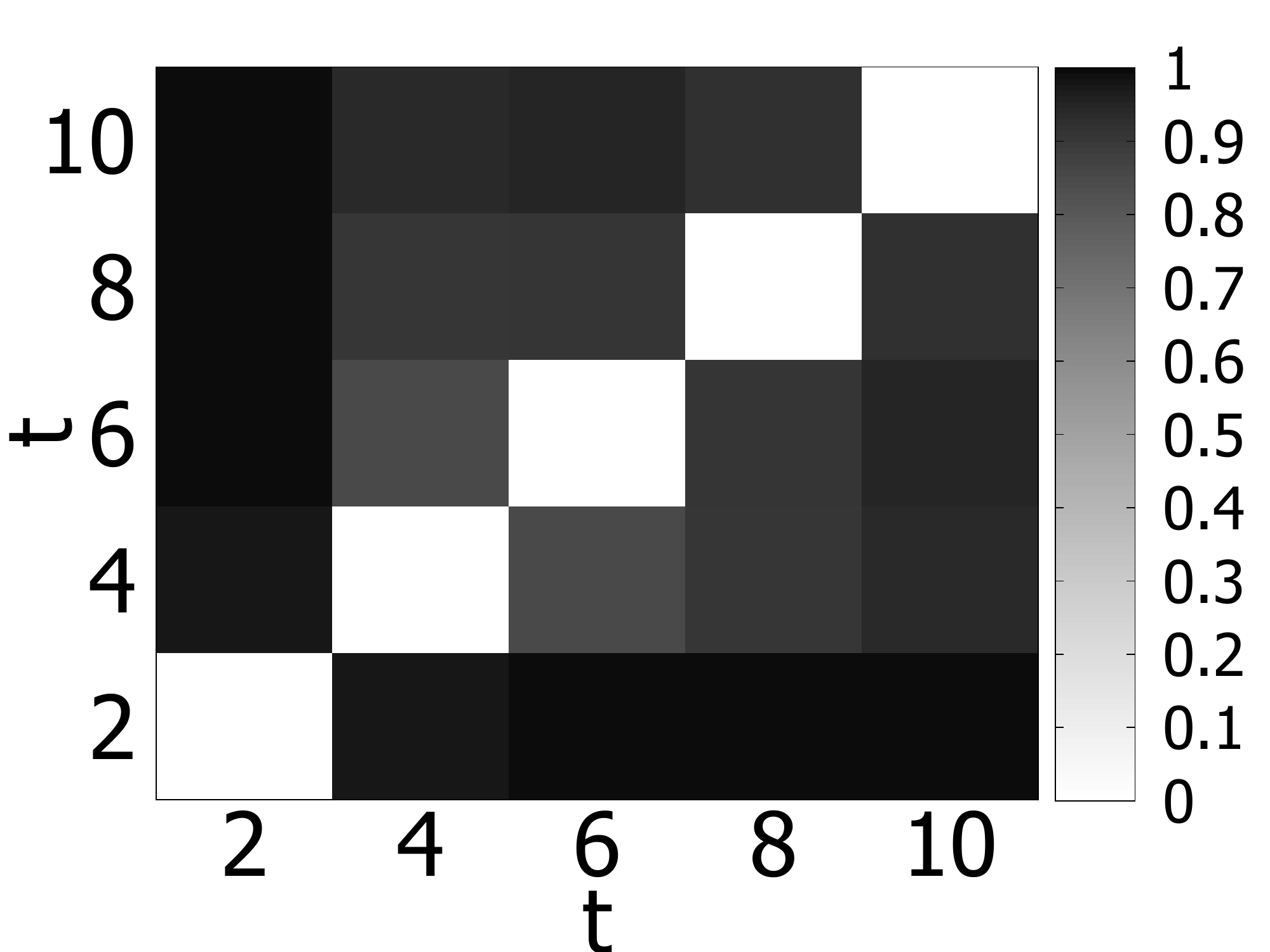}}
\subfigure[LiveJournal]{\includegraphics[width=0.25\linewidth]{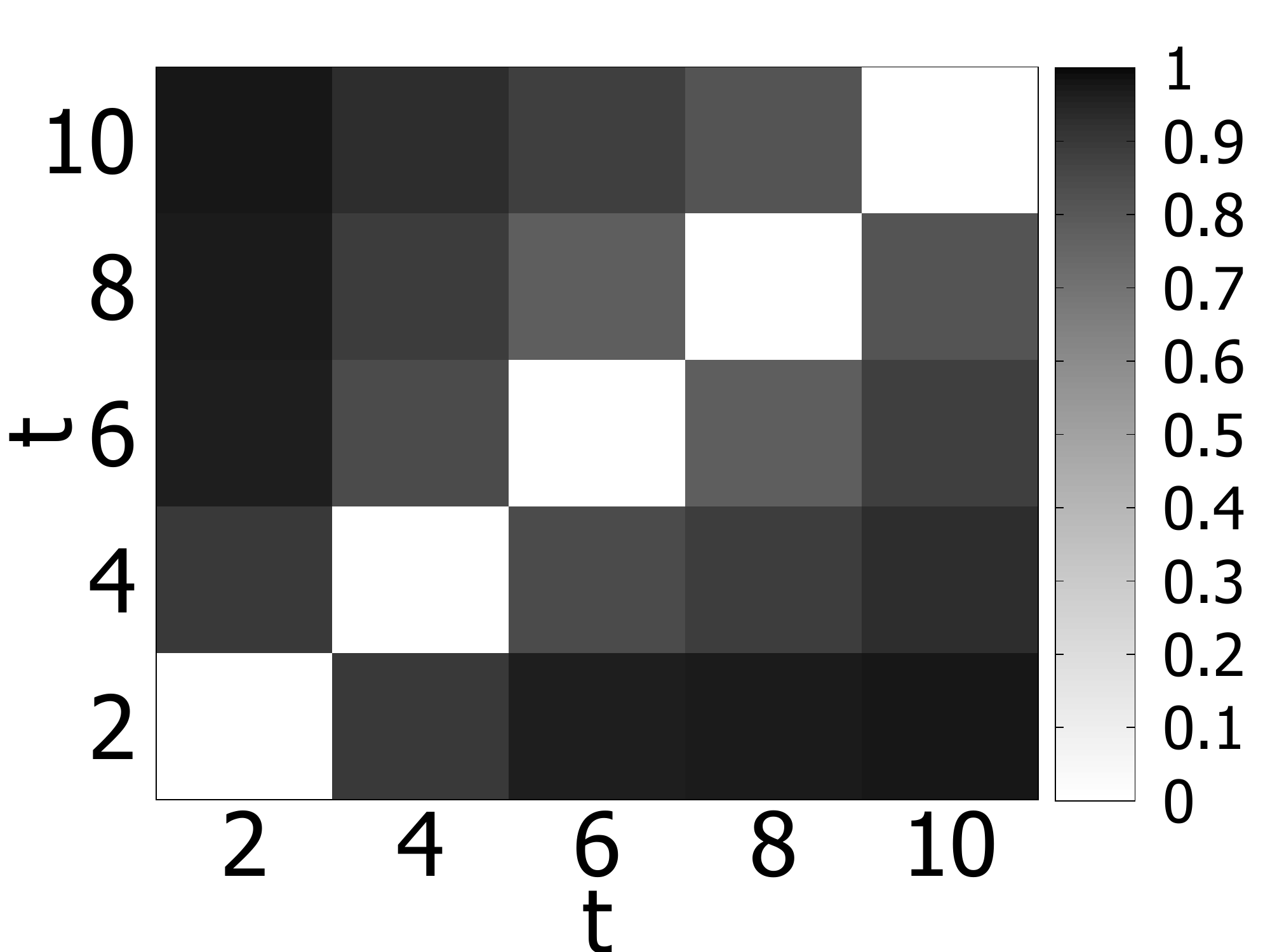}}
\subfigure[Orkut]{\includegraphics[width=0.25\linewidth]{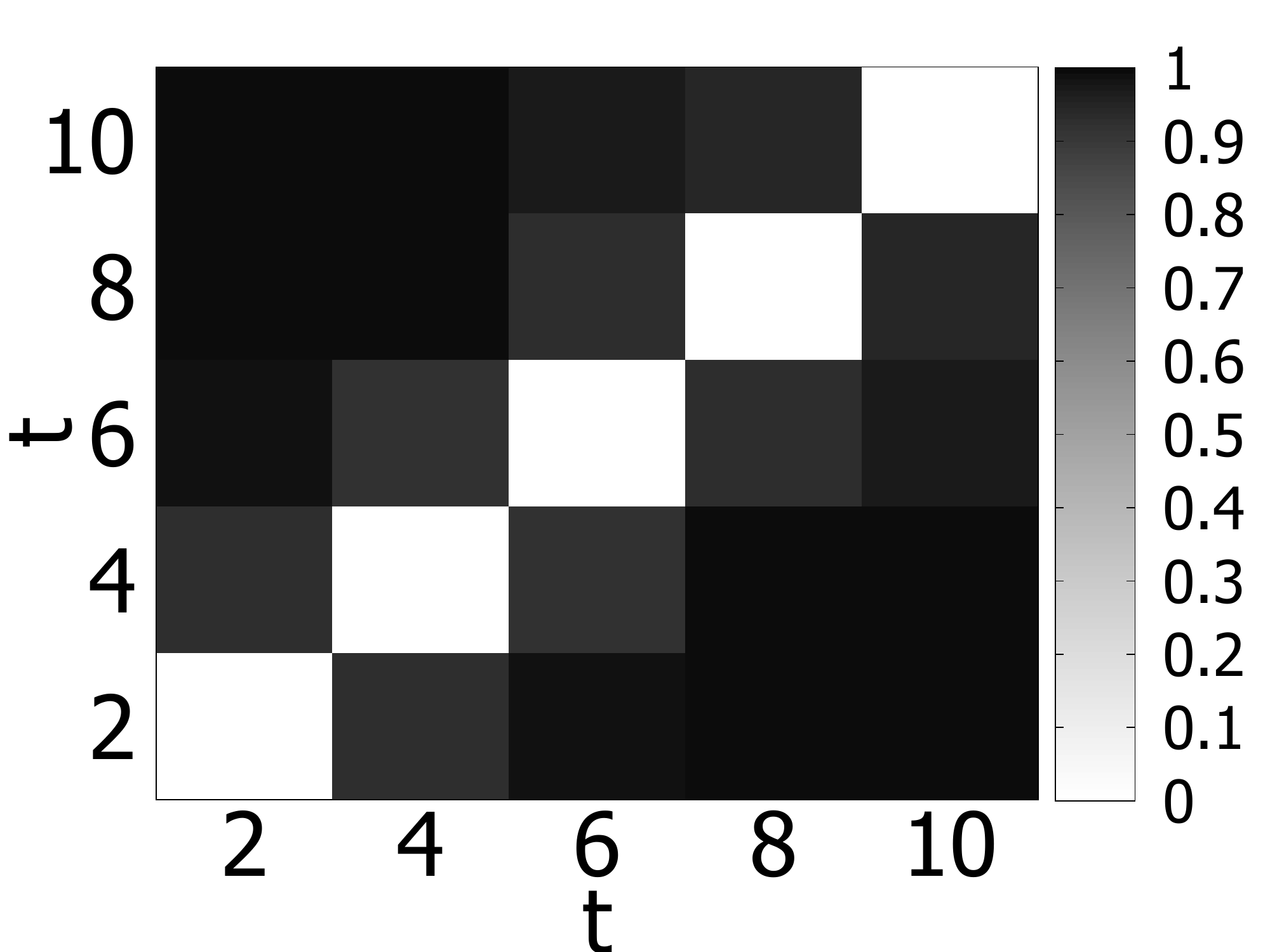}} 
 }
  \vspace*{-0.6cm}
\caption{Sensitivity heat matrices of \tcore model on all datasets. Each matrix element represents the Kendall's Tau distance between two top-100 ranking lists by \tcore model with different $t$.}
\vspace*{-0.4cm}
\label{fig.kendalltau}
\end{figure*}

\subsection{Effectiveness Evaluation}
In this experiment, we evaluate the effectiveness of our proposed $h$-index based structural diversity. We compare our method \hcore with state-of-the-art \tcore~\cite{HuangCLQY15} in the task of social contagion. Specifically, we adopt the independent cascade model to simulate the influence propagation process in graphs~\cite{GoyalLL11}. Influential probability of each edge is set to $0.01$. Then, we select 50 vertices as activated seeds by an influence maximization algorithm~\cite{TangSX15}. We perform 1000 times of Monte Carlos sampling for  propagation. For comparison, we count the number of activated vertices in the top-$k$ results by \tcore and \hcore methods. The method that achieves the largest number of activated vertices is regarded as the winner.  

First, we report the average activated rate by \hcore and \tcore method on  all four datasets in Fig.~\ref{fig.propagate}(a). Let  $D=$\{``Gowalla'',``Youtube'', ``LiveJournal'',``Orkut''\}. Given a dataset $d\in D$, the activated rate is defined as $f_k(d)=\frac{ActNum_k}{k}$, where $ActNum_k$ is the number of activated vertices in the top-$k$ result. The average activated rate is defined as $ActRate_k=\frac{\sum_{d\in D}f_k(d)}{|D|}$. Fig.~\ref{fig.propagate}(a) shows that our method \hcore achieves the highest activated rates, which significantly outperforms \tcore method for all different $t$. 
It indicates that the top-$k$ results found by \hcore tend to have higher probability to be affected in social contagion. 


\begin{figure*}[t]
\vspace*{-0.3cm}
\centering \mbox{
\subfigure[Average Activated Rate]{\includegraphics[width=0.3\linewidth]{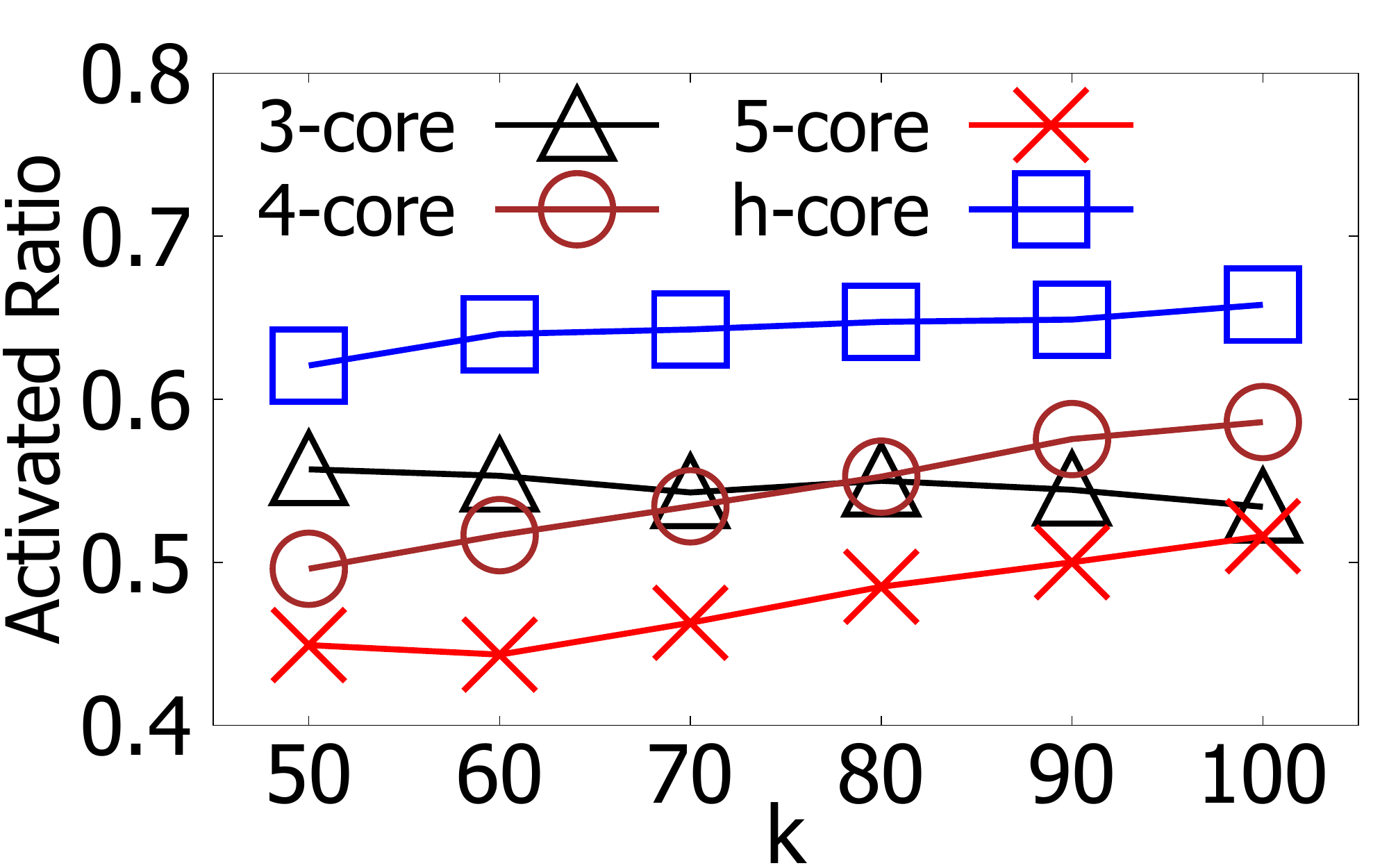}}
\hspace*{0.6cm}
\subfigure[Win Cases Distribution]{\includegraphics[width=0.3\linewidth]{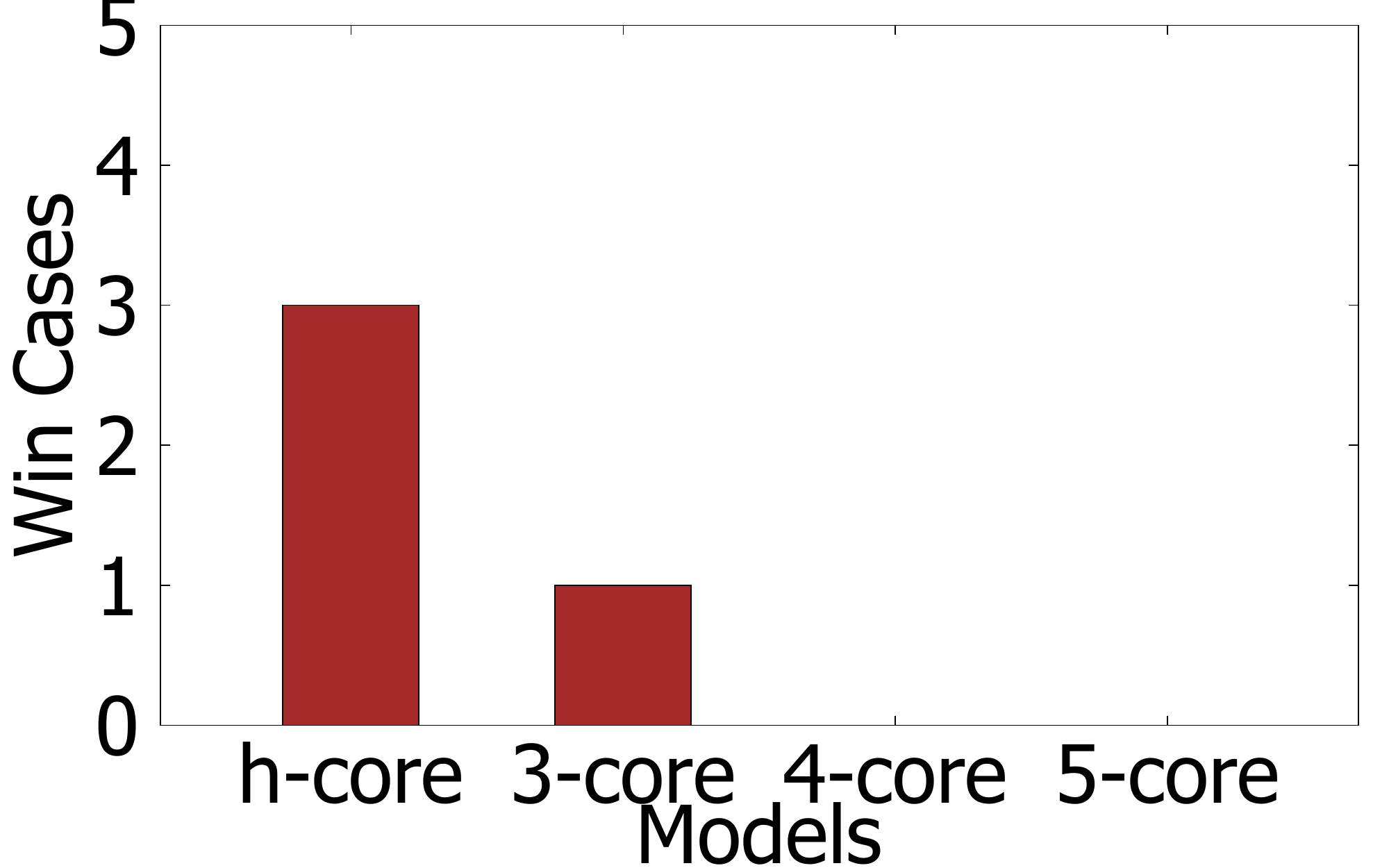}} }
 \vspace*{-0.4cm}
\caption{Comparison of \tcore and \hcore in terms of the average activated ratio and win cases on four datasets.}
\vspace*{-0.6cm}
\label{fig.propagate}
\end{figure*}

In addition, we also report the win cases of  \hcore and \tcore with different parameter $t$ on all dataset. We vary $t=\{2,3,4\}$ and set $k=100$ for all methods. The winner of a dataset is the method that achieves the highest number of activated vertices in this dataset. Fig.~\ref{fig.propagate}(b) shows the win cases of \tcore and \hcore. As we can see, \hcore wins on three datasets, which achieves the best performance. It further shows the superiority of our $h$-index structural diversity model. Besides, 3-core wins once, 4-core and 5-core win none, indicating that \tcore performs sensitively to parameter $t$.